\newif\iffull 

\makeatletter \@input{texdirectives.tex} \makeatother

\documentclass[acmsmall,screen]{acmart}

\setcopyright{rightsretained}
\acmPrice{}
\acmDOI{}
\acmYear{2020}
\copyrightyear{2020}
\acmJournal{JACM}
\acmVolume{1}
\acmNumber{1}
\acmArticle{1}
\acmMonth{1}

\iffull
\makeatletter
\def\@authorsaddresses{}
\fancypagestyle{firstpagestyle}{%
  \fancyhf{}%
  \renewcommand{\headrulewidth}{\z@}%
  \renewcommand{\footrulewidth}{\z@}%
}
\fancypagestyle{standardpagestyle}{%
  \fancyhf{}%
  \renewcommand{\headrulewidth}{\z@}%
  \renewcommand{\footrulewidth}{\z@}%
    \fancyhead[RO]{\@headfootfont\thepage}%
    \fancyhead[RE]{\@headfootfont\@shortauthors}%
    \fancyfoot[RO,LE]{}%
}
\def\@mkbibcitation{}
\makeatother
\renewcommand\footnotetextcopyrightpermission[1]{}
\fi

\usepackage[T1]{fontenc} %
\usepackage[utf8x]{inputenc}

\usepackage{color,my_framed,subfig}
\usepackage{agda-syntax}
\definecolor{shadecolor}{gray}{0.93}
\usepackage{hyperref}
\usepackage{anyfontsize} 
\usepackage{amsmath}     

\usepackage{amssymb}     
\usepackage{amsbsy}      
\usepackage{semantic}    
\usepackage{stmaryrd}   
\usepackage{mathpartir}  
\usepackage{color}
\usepackage{braket}      
\usepackage{mathtools}   
\usepackage{thmtools}    
\usepackage{marvosym} 
\usepackage{wasysym}     
\usepackage{balance}     
\usepackage{xspace}
\usepackage{mathrsfs} 
\usepackage{proof}
\usepackage{rotating}
\usepackage{tikz-cd}
\usepackage{xcolor}
\usepackage{changepage}
\usepackage{centernot}
\usepackage{comment}
\usepackage{listings}
\usepackage{mdframed}
\usepackage{centernot}
\usepackage{mathtools}
\usepackage{multirow}
\usepackage{xparse}
 
\usepackage{scalerel}[2016/12/29]

\definecolor{darkblue}{rgb}{0.0, 0.0, 0.55}
\definecolor{darkgreen}{rgb}{0.0, 0.55, 0.13}
\definecolor{darkred}{rgb}{0.55, 0.0, 0.0}
\definecolor{darkviolet}{rgb}{0.58, 0.0, 0.83}
\definecolor{lightblue}{rgb}{0.68, 0.85, 0.9}
\usepackage{lstcoq}

\lstnewenvironment{caml}{\lstset{language=caml}}{}

\newcommand{\myplus}{\oplus}
\makeatletter
\newcommand{\xMapsto}[2][]{\ext@arrow 0599{\Mapstofill@}{#1}{#2}}
\def\Mapstofill@{\arrowfill@{\Mapstochar\Relbar}\Relbar\Rightarrow}
\makeatother

\definecolor{propcolor}{HTML}{3F7D31}
\definecolor{mypurple}{HTML}{5B069D}

\definecolor{lightgray}{gray}{0.90}

\mathlig{[]}{\square}
\mathlig{|-s}{\vdash^{\!\forall}}
\mathlig{|-}{\vdash}
\mathlig{|}{\mid} 
\mathlig{@>}{\leadsto} 

\reservestyle{\oblang}{\mathsf}
\oblang{if[if\;],then[\;then\;],else[\;else\;],let[let\;],in[\;in\;]}

\newcommand{\code}[1]{\texttt{#1}}

\NewDocumentCommand{\List}{g}{\code{List}\IfValueT{#1}{~#1}}

\newcommand{\oblset}[1]{\textsc{#1}}

\newcommand{\Type}{\mathtt{Type}}
\newcommand{\Prop}{\oblset{Prop}}

\newcommand{\lx}{\ell} 

\newcommand{\iappev}[3][{\evl[i]}]{\mathbin{@^{#2}_{#1}}}

\newcommand{\ie}{\emph{i.e.,}\xspace}
\newcommand{\etal}{\emph{et al.}\xspace}
\newcommand{\eg}{\emph{e.g.,}\xspace}

\newcommand{\ev}{\varepsilon}
\newcommand{\evl}[1][]{{\color{mypurple}{\ev_{\lx #1}}}}

\newcommand{\assignev}[3][{\evl[i]}]{#2 \overset{#1}{:=} #3}

\newcommand{\opt}[1]{}

\mathlig{++}{\mathbin{++}}

\newmdenv[topline=false,bottomline=false,leftline=false]{borderright}

\newcommand\gCtx{\Xi}
\newcommand{\trans}[1]{\ensuremath{[#1]}}
\newcommand{\param}[2]{\transType{#1}_p^{#2}}
\newcommand{\uparam}[1]{\trans{#1}_u}
\newcommand{\transType}[1]{\ensuremath{[\![#1]\!]}}
\newcommand{\uparamT}[1]{\transType{#1}_u}
\newcommand{\uparamWB}[2]{\trans{#1}_u^{#2}}
\newcommand{\uparamTWB}[2]{\transType{#1}_u^{#2}}
\newcommand{\uparamEq}[1]{\transType{#1}_u^{eq}}
\newcommand{\uparamCoh}[1]{\transType{#1}_u^{coh}}
\newcommand{\minus}{$-$}

\newcommand\filledcirc{\scaleobj{0.68}{\otimes}}

\newcommand{\primeTrans}[1]{{#1}'}
\newcommand{\epsTrans}[1]{{#1}^{\varepsilon}}
\newcommand{\fstTrans}[1]{{#1}^{\circ}}
\newcommand{\sndTrans}[1]{{#1}^{\bullet}}
\newcommand{\RTrans}[1]{{#1}^{\filledcirc}}
\newcommand{\contextFst}[1]{|\! \! #1 \!\!|_{\circ}}
\newcommand{\contextSnd}[1]{|\!\! #1 \!\!|_{\bullet}}
\newcommand{\contextR}[1]{|\!\! #1 \!\!|_\varepsilon}
\newcommand{\contextAll}[1]{|\! \! #1 \!\!|}

\newcommand{\primeTransExt}[2]{[#1]_{_\square}^{#2}}

\newcommand{\defeq}{\triangleq}

\newcommand{\Nat}{\mathbb{N}}
\newcommand{\Bin}{\mathtt{Bin}}
\newcommand\NO{\mathtt{O}_\Bin}
\newcommand\NS{\mathtt{S}_\Bin}
\newcommand\RS{\RTrans{\mathtt{S}}}
\newcommand\RO{\RTrans{\mathtt{O}}}
\newcommand\Succ{\mathtt{S}}
\newcommand\Zero{\mathtt{O}}

\newcommand{\proji}[1]{#1.1}
\newcommand{\projii}[1]{#1.2}
\newcommand{\projiii}[1]{#1.3}
\newcommand{\Eqr}{\simeq}
\newcommand{\transEq}[2]{#1 \mathop{\#} #2}
\newcommand{\Eq}[2]{#1 \Eqr #2}

\newcommand{\heq}[2]{#1 \approx #2}

\newcommand{\heteqr}[2]{#1 \bowtie #2}
\newcommand{\heteqrP}[2]{#1 \bowtie_p #2}
\newcommand{\heteq}[4]{\heq{#1}{#2} : \heteqr{#3}{#4}}
\newcommand{\heteqP}[4]{\heqP{#1}{#2} : \heteqrP{#3}{#4} }

\newcommand{\heqP}[2]{#1 \approx_p #2}
\newcommand{\Univ}[1]{\RTrans{#1}}
\newcommand{\univTerm}[1]{\textsf{univ}_{#1}}

\newcommand{\equivPi}{\mathrm{Equiv}_\Pi}

\newcommand{\CC}{${\mathrm{CC}}_{\omega }$\xspace}
\newcommand{\Coq}{${\mathrm{Coq}}$\xspace}
\newcommand{\OCaml}{${\mathrm{OCaml}}$\xspace}
\newcommand{\CIC}{${\mathrm{CIC}}$\xspace}
\newcommand{\CICU}{${\mathrm{CIC}}_u$\xspace}
\newcommand{\vdashu}{\vdash_u}

\usepackage{booktabs} 
\usepackage{bussproofs}

\newcommand{\id}[1]{\text{id}_{#1}}

\newcommand{\Transp}[1]{\mathrel{\uparrow\!\!#1}}
\newcommand{\Transpe}[2]{\mathrel{\uparrow_{\!#1}\!{#2}}}
\newcommand{\typeOfI}{\mathtt{T}}
\newcommand{\typeOfci}{\mathtt{T_{c_i}}}
\newcommand{\typeOfrect}{\mathtt{T_{rect}}}

\newcommand{\UR}[1]{U\!R_{#1}}
\newcommand{\ur}[1]{r_{#1}}

\usepackage[ruled]{algorithm2e} 

\SetAlFnt{\small}
\SetAlCapFnt{\small}
\SetAlCapNameFnt{\small}
\SetAlCapHSkip{0pt}
\IncMargin{-\parindent}



\begin{document}

\renewcommand{\footnotemark}{\mbox{}}

\title{The Marriage of Univalence and Parametricity}
\subtitle{To appear in Journal of the ACM (accepted Oct. 2020)}

\author{Nicolas Tabareau}
\affiliation{%
  \institution{Gallinette Project-Team, Inria}
  \city{Nantes}
  \country{France}
}
\author{\'Eric Tanter}
\affiliation{%
  \institution{Computer Science Department (DCC), University of Chile}
  \city{Santiago}
  \country{Chile}
}
\author{Matthieu Sozeau}
\affiliation{%
  \institution{Gallinette Project-Team, Inria}
  \city{Nantes}
  \country{France}
}
\titlenote{This work is partially funded by ANID FONDECYT Regular
  Project 1190058, ANID/CONICYT REDES Project 170067, ERC Starting Grant
  CoqHoTT 637339 and Inria {\'E}quipe Associ{\'e}e GECO.}

\begin{abstract}
Reasoning modulo equivalences is natural for everyone, including mathematicians. Unfortunately, in proof assistants based on type theory, which are frequently used to mechanize mathematical results and carry out program verification efforts, equality is appallingly syntactic and, as a result, exploiting equivalences is cumbersome at best. Parametricity and univalence are two major concepts that have been explored in the literature to transport programs and proofs across type equivalences, but they fall short of achieving seamless, automatic transport. This work first clarifies the limitations of these two concepts when considered in isolation, and then devises a fruitful marriage between both. The resulting concept, called {\em univalent parametricity}, is an extension of parametricity strengthened with univalence that fully realizes programming and proving modulo equivalences. Our approach handles both type and term dependency, as well as type-level computation. 
In addition to the theory of univalent parametricity, 
we present a lightweight framework implemented in the Coq proof assistant that allows the user to transparently transfer definitions and theorems for a type to an equivalent one, as if they were equal. For instance, this makes it possible to conveniently switch between an easy-to-reason-about representation and a computationally-efficient representation, as soon as they are proven equivalent. The combination of parametricity and univalence supports {\em transport {\`a} la carte}: basic univalent transport, which stems from a type equivalence, can be complemented with additional proofs of equivalences between functions over these types, in order to be able to transport more programs and proofs, as well as to yield more efficient terms.
We illustrate the use of univalent parametricity on several examples, including a recent integration of native integers in \Coq. This work paves the way to easier-to-use proof assistants by supporting seamless programming and proving modulo equivalences.
\end{abstract}

\maketitle

\section{Introduction}
\label{sec:intro}

If mathematics is the art of giving the same name to different things, programming is the art of computing the same thing with different means. That sameness notion ought to be equivalence.
Unfortunately, in programming languages as well as proof assistants such as \Coq~\cite{Coq:manual} and Agda~\cite{agdaPaper}, the
notion of sameness or equality is appallingly syntactic. In
dependently-typed languages that also serve as proof assistants,
equivalences can be stated and manually exploited, but they cannot be
used as transparently and conveniently as syntactic or propositional equality.
The benefits we ought to get from having equivalence as the primary
notion of sameness include the possibility to state and prove results
about a data structure (or mathematical object) that is convenient to
formally reason about, and then automatically transport these results to
other structures, for instance ones that are computationally more
efficient, albeit less convenient to reason about.

Let us consider two equivalent representations of natural numbers available in the \Coq proof assistant (Figure~\ref{fig:nat-binnat}): Peano natural numbers \coqe{nat}, with constructors \coqe{O} and \coqe{S}, and binary natural numbers \coqe{N}, which denote a sequence of bits with a leading 1. Defining functions over \coqe{nat} and reasoning about them is simple. For instance, 
\coqe{+_nat : nat -> nat -> nat} is a simple induction on the first argument, and proving that addition is commutative is similarly direct. Conversely, addition on binary natural numbers \coqe{+_N : N -> N -> N} is defined with three functions---two mutually-recursive functions on \coqe{positive} and a simple function---making most reasoning much more involved. The other side of the comparison is that computing with \coqe{nat} is much less efficient (if at all possible!) than computing with \coqe{N}.
Ideally, one would want to apply easy inductive reasoning on \coqe{nat} to establish properties of efficient functions defined on \coqe{N}.

\begin{figure}[t]
  \hspace{-2em}
  \begin{tabular}{ccc}
    \begin{coq}
          Inductive nat : Set :=
           | O : nat
           | S : nat -> nat
           
         \end{coq}
      &
    \begin{coq}
      Inductive N : Set :=
       | NO : N
       | Npos : positive -> N
       
    \end{coq}
    &
      \begin{coq}
      Inductive positive : Set :=
       | xI : positive -> positive
       | xO : positive -> positive
       | xH : positive
     \end{coq}
       \end{tabular}
  \caption{Definition of \coqe{nat} and \coqe{N} in \Coq}
  \label{fig:nat-binnat}
\end{figure}

\paragraph{The Challenge of Automatic Transport} An equivalence between \coqe{nat} and \coqe{N} consists of two transport functions \coqe{bupa  : N -> nat} and \coqe{bdna  : nat -> N} together with the proof that they are inverse of each other. Manually exploiting such an equivalence in order to transport properties on \coqe{nat} to properties on \coqe{N} is however challenging, even for simple properties. 

Consider the commutativity of addition on \coqe{nat}:

\begin{shaded}
\begin{coq}
Definition +_nat_comm : forall (n m : nat), n +_nat m = m +_nat n. (* simple inductive proof *)
\end{coq}
\end{shaded}

from which one would like to deduce the commutativity of addition on \coqe{N}:
\begin{shaded}
\begin{coq}
Definition +_N_comm : forall (n m : N). n +_N m = m +_N n.
\end{coq}
\end{shaded}

A proof of \coqe{+_N_comm} that exploits \coqe{+_nat_comm} and the \coqe{nat}-\coqe{N} equivalence would proceed as follows:
\begin{shaded}
\begin{coq}
1. (bupa n) +_nat (bupa m) = (bupa m) +_nat (bupa n)  (* by +_nat_comm (bupa n) (bupa m) *)
2. bdna ((bupa n) +_nat (bupa m)) = bdna ((bupa m) +_nat (bupa n)) (* by congruence *)
3. (bdna bupa n) +_N (bdna bupa m) = (bdna bupa m) +_N (bdna bupa n) (* bdna is a monoid homomorphism *)
4. n +_N m = m +_N n (* by equivalence *)
\end{coq}\end{shaded}
Observe how one is forced to explicitly rewrite and reason about
transports at each step. In particular, step 3 requires to show that
\coqe{bdna} is a monoid homomorphism between \coqe{+_nat} and \coqe{+_N}.
This does not follow from the type
equivalence between \coqe{nat} and \coqe{N}, and therefore needs to be
manually proven.

From this simple example, it is easy to imagine the difficulty of transporting entire libraries of structures and lemmas about their properties. 
The promise of {\em automatic transport across type equivalences} is to seamlessly allow users to operate with the most-suited representation as needed. In particular, deriving \coqe{+_N_comm} ought to be as simple as transporting \coqe{+_nat_comm} (using a general transport operator \coqe{upa} whose source and target types are inferred from context):
\begin{shaded}
\begin{coq}
Definition +_N_comm : forall (n m : N). n +_N m = m +_N n := upa +_nat_comm.
\end{coq}
\end{shaded}
In the literature, two major concepts have been explored to achieve automatic transport across equivalences: {\em parametricity} and {\em univalence}. This article demonstrates that both of them are insufficient taken in isolation, and that it is possible to devise a marriage of univalence and parametricity that leverages both in order to fully realize programming and proving modulo equivalences. 

\paragraph{Parametricity.}

Since the seminal work of~\citet{magaudBertot:types2000} on translating proofs between different representations of natural numbers in \Coq, there has been a lot of work in this direction, motivated by both program verification and mechanized mathematics, with several libraries available for either Isabelle/HOL~\cite{huffmanKuncar:cpp2013} or \Coq~\cite{cohenAl:cpp2013,zimmermannHerbelin:arxiv2015}.
At their core, these approaches build on the notion of parametricity~\cite{Reynolds83} and its potential for free theorems about observational equivalences~\cite{Wadler89}, in order to obtain results such as data refinements for free~\cite{cohenAl:cpp2013} and proofs for free~\cite{bernardyAl:jfp2012}. 

Such a parametric transport is essentially a {\em white-box} approach that structurally rewrites observationally-equivalent terms. 
The previous example of \coqe{N_comm} can be handled by parametric transport.  However, as we will demonstrate, the approach does not fully apply in the dependently-typed setting where computation at the type level is essential. (The \coqe{N_comm} example luckily does not rely on any type-level computation.)

\paragraph{Univalence}
Univalence~\cite{voevodsky:cmu2010} is a novel principle for mathematics and type theory that postulates that equivalence is equivalent to equality.
Leaving aside the most profound mathematical implications of Homotopy Type Theory (HoTT) and univalence~\cite{hottbook}, this principle should fulfill the promise of automatic transport of programs, theorems, and proofs across equivalences. 
There are currently two major approaches to realize univalence in a type theory. In Martin-L{\"o}f Type Theory (MLTT)~\cite{MARTINLOF197573}, and related theories such as the Calculus of (Inductive) Constructions~\cite{coquandHuet:ic1988,Paulin15}, univalence can only be expressed as an {\em axiom}. However, by the Curry-Howard correspondence, axioms have no computational content, since they correspond to free variables. Therefore an axiomatic general univalent transport is not effective. In concrete terms, this means that using axiomatic univalent transport will yield a ``stuck term'', stuck at the use of the axiom.
There are several recent developments to build a dependent type theory with a computational account of univalence, most notably cubical type theories~\cite{cubicaltt,altenkirch15:towards,angiuliAl:csl2018}, and concrete implementations such as Cubical Agda~\cite{vezzosiAl:icfp2019} have started to appear. Cubical type theories fully achieve the consequences of realizing univalence, such as giving computational content to both functional and propositional extensionality.

Irrespective of how univalence is realized, 
{\em univalent transport} allows exploiting an equivalence between two types $A$ and $B$ in order to establish an equivalence between $P\;A$ and $P\;B$, for any arbitrary predicate $P$. But univalent transport alone does not address a major challenge for automatic transport, namely that of inferring, from basic equivalences, the common predicate $P$ out of arbitrarily complex dependent types. Additionally, while it is universally applicable as a {\em black-box} approach, univalent transport can yield unsatisfactory transported terms, as explained next.

\paragraph{Transport {\`a} la Carte}

With univalent transport, one can always convert any development that uses \coqe{nat} into one that uses \coqe{N}, both in computationally-relevant parts and in parts that deal with reasoning and formal properties. However, univalent transport does not necessarily reconcile ease of reasoning with efficient computation. 
Indeed, univalently transporting a function \coqe{nat -> nat} to a function \coqe{N -> N} yields a function that first converts its binary argument to a natural number, performs the original (slow) computation and finally converts the result back to a binary number. Dually, if one starts from a \coqe{N -> N} function and transports it to \coqe{nat -> nat}, the resulting function will still execute efficiently, but simple \coqe{nat}-based inductive reasoning will not be applicable to it.

The problem is that univalent transport across the \coqe{nat}-\coqe{N} equivalence does not magically exploit the correspondence between different {\em implementations} of functions that operate on these types, 
such as between \coqe{+_nat} and \coqe{+_N}. Such term-level correspondences {\em are} exploited in parametricity-based approaches, and require additional proof and engineering effort.

Therefore, in addition to addressing the limitations of parametricity and univalence when taken in isolation, an essential component of automatic transport is the tradeoff between the cost of manually establishing equivalences (between both types and functions that operate on them), and the ease of automatic univalent transport. One wishes for an automatic transport mechanism {\em {\`a} la carte}, which exploits user-provided equivalences between terms when available, and falls back to univalent transport otherwise.

\paragraph{Univalent Parametricity} This article proposes a marriage of univalence and parametricity that enables automatic transport {\`a la carte} across equivalences. It deeply connects and intertwines (white-box) parametric transport and (black-box) univalent transport in a fruitful manner.
Essentially, univalent parametricity is a strengthening of the
parametricity translation for dependent types that demands the
relation on the universe to be compatible with equivalences. 
This paper is structured as follows. We first recall parametricity in type theory and present its limitations to transport definitions across equivalences (\S\ref{sec:param-not-enough}). We then proceed similarly with univalence, highlighting the complementarity between both approaches (\S\ref{sec:univ-not-enough}). Next, we illustrate how univalent parametricity achieves seamless automatic transport across equivalences from a user point of view (\S\ref{sec:univ-param-acti}). We develop the theory of univalent parametricity for the Calculus of Constructions with universes \CC (\S\ref{sec:univparam}), and for the Calculus of Inductive Constructions CIC (\S\ref{sec:inductives}). We present the realization of univalent parametricity in the \Coq proof assistant as a shallow embedding that exploits the typeclass mechanism (\S\ref{sec:coq}). Note that this implementation in \Coq does not give any computational content to univalence and extensionality axioms; instead, it brings automatic univalent transport to programmers, using such axioms sparingly in order to disrupt computation as little as possible. We explain how the illustration of \S\ref{sec:univ-param-acti} is effectively implemented in \Coq (\S\ref{sec:extended-example}). We end by describing a case study related to the recent integration of native integers in \Coq (\S\ref{sec:ffi}). \S\ref{sec:related} discusses related work, and \S\ref{sec:conclu} concludes.

The complete \Coq development is available online:
\begin{center}
\url{https://github.com/coqhott/univalent_parametricity}
\end{center}

\noindent {\bf Prior publication.}
This paper is a substantial extension of a prior conference publication~\cite{tabareauAl:icfp2018}. First, we explain in details the limitations of both parametricity and univalence when considered in isolation (\S\ref{sec:param-not-enough}-\S\ref{sec:univ-not-enough}). Second, we extend our original proposal to integrate user-defined correspondences between terms of different (related) types, which is absolutely necessary to reconcile ease of reasoning and efficient computation, and to realize transport {\`a} la carte. The illustrations of \S\ref{sec:univ-param-acti} and \S\ref{sec:extended-example} are therefore novel as well, as they take advantage of this new feature, in addition to providing a detailed user perspective on the \Coq framework. \S\ref{sec:coq} is extended accordingly to deal with transport {\`a} la carte. \S\ref{sec:univparam} clarifies the two reasoning principles, white-box and black-box, supported by univalent parametricity, and explains how they support transport {\`a} la carte. Finally, the case study of reasoning about/with native integers (\S\ref{sec:ffi}) is entirely new.

\section{Parametricity is Not Enough}
\label{sec:param-not-enough}

We first review the development of parametricity in dependently-typed theories (\S\ref{sec:param-as-logic}), and discuss its use and limitations for transporting some programs and proofs (\S\ref{sec:param-lift}). Finally, we consider an extension of parametricity that addresses some limitation, but is still limited when type-level computation is involved (\S\ref{sec:het-param}).

\subsection{Parametricity for Dependent Types}
\label{sec:param-as-logic}

\begin{figure}[t]

\begin{align*}
\param{\Type_i}{} \ A \ B & \defeq  A -> B -> \Type_i \\[0em]
\param{\Pi a:A. B}{} \ f \ g & \defeq  \Pi (a:A) (\primeTrans{a}:\primeTrans{A}) (\epsTrans{a}:\param{A}{} \ a \
\primeTrans{a}). \param{B}{}~ (f\ a) ~ (g\ \primeTrans{a})\\[0em]
\param{x}{}& \defeq  \epsTrans{x} \\[0em]
\param{\lambda x:A. t}{} & \defeq  \lambda (x:A) (\primeTrans{x}:\primeTrans{A}) (\epsTrans{x}:\param{A}{}\ x\  
  \primeTrans{x}). \param{t}{}\\[0em]
\param{t \  u}{} & \defeq \param{t}{} \ u \ \primeTrans{u} \ \param{u}{} \\[1em]
\param{\cdot}{} & \defeq  \cdot \\
\param{\Gamma,x:A}{} & \defeq \param{\Gamma}{},
  x:A, \primeTrans{x}:\primeTrans{A}, \epsTrans{x} : \param{A}{} \ x \ \primeTrans{x}
\end{align*}
\caption{Parametricity translation for \CC (from~\cite{bernardyAl:jfp2012})}
\label{fig:param}
\end{figure}

Reynolds originally formulated the relational interpretation of
types to establish parametricity of System~F~\cite{Reynolds83}. 
More recently, 
\citet{bernardyAl:jfp2012} generalized the approach to pure type
systems, including the Calculus of
Constructions with universes \CC, and its extension with inductive
types, the Calculus of Inductive Constructions CIC, which is at the core
of proof assistants like \Coq.\footnote{\CC features a predicative
  hierarchy of universes ${\Type }_{i}$, and also an impredicative universe
  $\Prop$. In this paper, we focus on the predicative hierarchy, because adding an impredicative universe has little impact. Section~\ref{sec:prop} explains the minor changes for integrating the impredicative universe $\Prop$.}

The syntax of \CC includes a hierarchy of universes $\Type_i$, variables, applications, lambda expressions and dependent function types:
\begin{displaymath}
      {A, B, M, N}\mathrel{::=}{{\mbox{${\Type }_{i}$}}\mid {\mbox{$x$}}\mid {\mbox{$M\ N$}}\mid {\mbox{$\lambda x : A.\,M$}}\mid {\mbox{$\Pi x : A.\,B$}}}
\end{displaymath}
Its typing rules are standard, and hence omitted here---see \citet{Paulin15} for a recent presentation.

Parametricity for \CC can be defined as a logical relation
$\param{A}{}$ for every type $A$. Specifically, $\param{A}{}\ a_1\ a_2$ states that $a_1$ and $a_2$ are related at type $A$. 
The essence of \citeauthor{bernardyAl:jfp2012}'s approach is to express parametricity as a translation from terms to the expression of their relatedness {\em within} the same theory; indeed, the expressiveness of \CC allows the logical relation to be stated in \CC itself. Note that because terms and types live in the
same world, $\param{-}{}$ is defined for every term. 

Figure~\ref{fig:param} presents the definition of $\param{-}{}$ for \CC,
based on the work of \citet{bernardyAl:jfp2012}.
For the universe $\Type_i$, the translation is naturally defined as (arbitrary)
binary relations on types. 
For the dependent function type $\Pi a:A. B$, the translation specifies that related inputs at $A$, as witnessed by $e$, yield related outputs at $B$. 

Following \citeauthor{bernardyAl:jfp2012}, the prime notation (\eg~$\primeTrans{A}$) denotes duplication with renaming, where each free variable $x$ is replaced with $\primeTrans{x}$. 
Similarly, the translation of a lambda term $\lambda x:A. t$ is a function that takes two arguments and a witness $\epsTrans{x}$ that they are related; a variable $x$ is translated to $\epsTrans{x}$; a translated application passes the original argument, its renamed duplicate, along with its translation, which denotes the witness of its self-relatedness. The translation of type environments follows the same augmentation pattern, with duplication-renaming of each variable as well as the addition of the relational witness $\epsTrans{x}$.

Armed with this translation, it is possible to prove an abstraction theorem à la Reynolds, saying that a well-typed term is related to itself (more precisely, to its duplicated-renamed self):

\begin{theorem}[Abstraction theorem]
  \label{thm:abstraction-param}
  If $\Gamma \vdash t :A$ then 
  $\param{\Gamma}{} \vdash \param{t}{} : \param{A}{}\ t\ \primeTrans{t}$.
\end{theorem}
In particular, this means that the translation of a term  $\param{t}{}$ is itself the {\em proof} that $t$ is relationally parametric.

The abstraction theorem is proven by showing the fundamental
property of the logical relation for each constructor of the theory.
In particular, for the cumulative hierarchy of universes,
$\vdash \Type_i : \Type_{i+1}$. This means that we have a kind of
fixpoint property for the relation on $\Type_i$:
$$ 
\vdash \param{\Type_i}{} : \param{\Type_{i+1}}{} \ \Type_i \ \Type_i.
$$
For parametricity, this property holds because the following is a proof term:
$$ 
\lambda (A \  B : \Type_i) . \  \Type_i : \Type_i -> \Type_i -> \Type_{i+1}.
$$
Note that this necessary fixpoint property is not necessarily trivial to satisfy in any variant of parametricity, as we will see later (\S\ref{sec:univparam}).

\subsection{Using Parametricity to Transport Programs and Proofs}
\label{sec:param-lift}

The parametricity translation together with the abstraction theorem (Theorem~\ref{thm:abstraction-param}) are
powerful tools to derive free theorems (and proofs)~\cite{bernardyAl:jfp2012}. 
However, the abstraction theorem is only concerned with what we can
call {\em reflexive homogeneous} instances of the logical relation, \ie~relating a term with itself (\ie~reflexive) and hence at the same type (\ie~homogeneous). 
Thus, in order to be able to relate functions and theorems over
different types, such as \coqe{nat} and \coqe{N}, the standard
solution is to define functions manipulating a common abstraction of their
algebraic structure---in the case of natural numbers, a type with a zero and a
successor function---together with an elimination principle.
Then, by parametricity, we know that such a function defined on the common abstraction behaves the same if
we instantiate it with \coqe{nat} or \coqe{N}, because it must
preserve any relation between \coqe{nat} and \coqe{N}, in particular
equivalences. 
This is for instance the approach taken in the CoqEAL framework~\cite{cohenAl:cpp2013}.

Parametricity presents two important issues, which we call the {\em anticipation problem} and the
{\em computation problem}. The anticipation problem is that, in order to reap the benefits of parametricity to transport programs and proofs, one must anticipate and explicitly exhibit {\em a priori} an interface that is common to the types dealt
with, and to define all functions generically in terms of this common interface. Engineering-wise, this anticipation might be problematic. Furthermore, defining the right interface might be challenging.
Of course, in the case of the addition on natural
numbers\footnote{This function definition corresponds to the infix
  notation \coqe{+_nat} used in \S\ref{sec:intro}.}, it is fairly
straightforward to define the common interface, because the definition
of \coqe{plus} is only using the successor function and the
eliminator.
\begin{shaded}
\begin{coq}
  Definition plus (n m : nat) : nat := nat_rect (fun _ => nat) m (fun _ res => S res) n.
\end{coq}
\end{shaded}
Similarly, as already noticed by \citet{cohenAl:cpp2013}, finding the right interface is direct when dealing with
primitive inductive types, but it becomes quite challenging when
dealing with types defined using a combination of several type
constructors.

The computation problem is that  parametricity does not scale to computation at the type level.
To illustrate this, consider the proof that \coqe{O} is different from \coqe{S n}, 
for every natural number \coqe{n}.
\begin{shaded}
\begin{coq}
Definition diff n (e : O = S n) : False :=
  let P_nat := nat_rect (fun _ => Type) (O = O) (fun n _ => False) in
  eq_rect nat O (fun n' _ => P_nat n') (eq_refl O) (S n) e.
\end{coq}
\end{shaded}
This definition uses the elimination principle of equality over a
predicate that is defined by elimination on natural numbers.
It typechecks because in the branch for
\coqe{O}, \coqe{P_nat O} reduces to \coqe{O = O}
and in the branch \coqe{S n, e}, \coqe{P_nat (S n)} reduces to \coqe{False}.

Now, if one tries to generalize this definition of this function by
making it modular with respect to any type of natural numbers with zero,
successor and a constant for the elimination principle, the result is ill-typed,
because the abstraction \coqe{P_abs} of \coqe{P_nat} does not compute
and so \coqe{P_abs O} is not {\em definitionally} equal to
\coqe{O = O}.
This issue can be sidestepped by adding the computational laws of the
eliminator on natural numbers as \emph{propositional} equalities in
the generalized version. But then, one needs to deal with rewriting
explicitly where otherwise everything was handled implicitly by
conversion.\footnote{Note that this issue appears because we are
  working with an intensional type theory. It would not be present in
  an extensional type theory, but in this work we only consider theories with a
  decidable type checking algorithm.}
This rewriting phase is not at all handled by parametricity.

\subsection{Heterogeneous Parametricity Translation}
\label{sec:het-param}

To address the anticipation problem described above, we would like to be able to relate \coqe{nat} or \coqe{N} {\em directly}---\ie~without relying on a common interface that captures their algebraic structure---simply because they are equivalent as types.

Observe that the definition of the parametricity translation of
\citeauthor{bernardyAl:jfp2012} given in Figure~\ref{fig:param} is {\em homogeneous}, in that terms are related {\em at the same type}, 
\ie~$\param{A}{}\ a_1\ a_2$. This allows us to provide instances of the parametricity relation such as $\param{\Type_i}{}$ \coqe{nat} \coqe{N}. 
But once \coqe{nat} and \coqe{N} are related as types, we
will want to relate some of their inhabitants, such as \coqe{O} and
\coqe{NO}, which means we also need to consider {\em heterogeneous}
instances, \ie~over terms of different (related) types.

But actually, parametricity itself is eminently heterogeneous: modifying 
the parametricity translation to reflect this is, in fact, straightforward. 
It suffices to additionally consider a global context~$\gCtx$ of
defined constant triples, where each triple consists of two constants 
(such as \coqe{O} and \coqe{NO})
and a witness that they are parametrically related. 
The global context $\gCtx$ is defined as the following telescope $\gCtx_n$:
\begin{align*}
\gCtx_0 & = \cdot \\
\gCtx_1 & = (\fstTrans{c}_1 : \fstTrans{A}_1 ;\ \sndTrans{c}_1 :\sndTrans{A}_1 ;\ \RTrans{c}_1 : 
\param{A_1}{\gCtx_0}\ \fstTrans{c}_1\ \sndTrans{c_1}) \\
\ldots & \\
\gCtx_n & = \gCtx_{n-1}, 
            (\fstTrans{c}_n : \fstTrans{A}_n ;\ \sndTrans{c}_n : \sndTrans{A}_n ;\ \RTrans{c}_n : \param{A_n}{\gCtx_{n-1}}\ \fstTrans{c}_n\ \sndTrans{c}_n)
\end{align*}

Note that in the definition above, we have extended
the parametricity translation to additionally take the global context into account, 
$\param{\cdot}{\gCtx}$. 
The definition of
Figure~\ref{fig:param} is accordingly extended on constants as follows:
$$\param{\fstTrans{c}}{\gCtx} \defeq \RTrans{c} \text{ when }
(\fstTrans{c} :\_\ ;\ \sndTrans{c}:\_\ ;\ \RTrans{c}: \_) \in \gCtx$$
We note $\contextFst{\gCtx}$, $\contextSnd{\gCtx}$, and $\contextR{\gCtx}$, the typing contexts obtained by projecting the respective components of each triple of $\gCtx$, and $\contextAll{\gCtx}$
the whole typing context $\contextFst{\gCtx}, \contextSnd{\gCtx}, \contextR{\gCtx}$.

To state the fundamental theorem of parametricity in this setting, we
need to define a notion of \emph{white box} translation of a given term,
$\primeTransExt{a}{\gCtx}$, which is essentially the identity function except
for constants in $\gCtx$, which are translated as
$$\primeTransExt{\fstTrans{c}}{\gCtx} \defeq \sndTrans{c} \text{ when }
(\fstTrans{c} :\_\ ;\ \sndTrans{c}:\_\ ;\ \RTrans{c}: \_) \in \gCtx$$
We can now state the ``White Box'' Fundamental Property (FP): 
when a term $a$ is well-typed in the context
$\contextFst{\gCtx}$, its white box translation
$\primeTransExt{a}{\gCtx}$ is also well-typed in the context
$\contextSnd{\gCtx}$, and the parametricity translation of $a$ provides
a witness that $a$ is related
to $\primeTransExt{a}{\gCtx}$.
\begin{corollary}[White Box Fundamental property]
  \label{white box param}
  \label{th:wbfp-intro}
  If $\contextFst{\gCtx}~\vdash a:A$ then
  $\contextSnd{\gCtx}~\vdash \primeTransExt{a}{\gCtx}: \primeTransExt{A}{\gCtx}$ and
  $\contextAll{\gCtx}~\vdash \param{a}{\gCtx} : \param{A}{\gCtx}\ a\ \primeTransExt{a}{\gCtx}$.
\end{corollary}

\begin{proof}
  From a (global) context $\contextAll{\gCtx}$, one can construct a
  substitution $\sigma$ from  $\contextAll{\gCtx}$ to
  $\param{\contextFst{\gCtx}}{}$
  by associating the triple
  $(\fstTrans{c} , \sndTrans{c} , \RTrans{c})  \mapsto
  (\fstTrans{c}, \primeTrans{\fstTrans{c}}, 
  \epsTrans{\fstTrans{c}})$. We note $\sigma(t)$ the application of
  the substitution $\sigma$ to a term $t$.
  By the abstraction theorem, we have that
  $$
  \param{\contextFst{\gCtx}}{}~\vdash \param{a}{} :
  \param{A}{}\ a\ \primeTrans{a}.
  $$
  The corollary follows from the fact that
  $ \sigma(\param{a}{}) \equiv \param{a}{\gCtx} $ and 
  $ \sigma(\primeTrans{a}) \equiv \primeTransExt{a}{\gCtx} $.
\end{proof}

Note that this property is only valid in a closed world (\ie no variables
but potentially global constants), because, as parametricity
is not internalized in the theory, there is no witness that variables or
(equivalently) axioms are parametric.

The introduction of the global context $\gCtx$ allows us to provide a direct
heterogeneous extension to the parametricity translation, 
however mentioning it explicitly in 
the translation makes the notation heavy. 
To avoid mentioning the global context explicitly, we introduce the
notation
$$\heteqP{a}{b}{A}{B} \defeq \param{A}{\gCtx}~a~b
$$
which relates two terms $a$ and $b$ at
two related---but potentially different---types $A$ and $B$.
Of course, this notation only makes sense when
$
B = \primeTransExt{A}{\gCtx}
$
in the current global context $\gCtx$, and we implicitly assume it is
the case when we use this notation.
The definition of the parametricity translation of
Figure~\ref{fig:param} for closed terms is the special homogeneous
case, where there is no global context and where terms are related at
the same type, \ie~$\param{A}\ a_1\ a_2$ corresponds to
$\heteqP{a_1}{a_2}{A}{A}$. 

The heterogeneous version of the parametricity translation makes it
possible to relate terms of different types, such as
$\heteqP{\text{\coqe{O}}}{\text{\coqe{NO}}}{\text{\coqe{nat}}}{\text{\coqe{N}}}$,
assuming that
$\heteqP{\text{\coqe{nat}}}{\text{\coqe{N}}}{\Type}{\Type}$ appears in
the global context.
Hereafter, whenever  $\heteqP{a}{b}{A}{B}$ holds, we say that $a$ and $b$ are
{\em parametrically related}. 
With the notation that makes the global context implicit, the
parametricity relation on dependent functions can be expressed as:
$$
\heteqP{f}{g}{\Pi a:A. P~a}{\Pi b:B. Q~b} \defeq \Pi \ (a:A) \
(b:B). \ \heteqP{a}{b}{A}{B} -> \heteqP{f~a}{g~b}{P~a}{Q~b}
$$
assuming $\heteqP{A}{B}{\Type}{\Type}$ and
$\heteqP{P}{Q}{A -> \Type}{B -> \Type}$.

For instance, to show that \coqe{nat} and \coqe{N} are parametrically
related, one needs to provide a relation \coqe{R_nat_N} between
\coqe{nat} and \coqe{N}.
While there are several equivalent ways of defining this relation,
the canonical one reuses the transport function \coqe{bupa}\; 
that comes from the equivalence between \coqe{nat} and \coqe{N}:
\begin{shaded}
\begin{coq}
  Definition R_nat_N := fun n m => n = bupa m.
\end{coq}
\end{shaded}

Then, in order to make explicit that \coqe{N} behaves the same as
\coqe{nat}, one can define a successor function
\coqe{NS : N -> N := fun n => n + 1} and show that \coqe{O} and \coqe{NO} are
parametrically related in the global context which contains
\coqe{(nat,N,R_nat_N)}, as well as \coqe{S} and \coqe{NS}, which
amounts to providing inhabitants for the following types:
$$
\begin{array}{rcl}
  \RO : \heteqP{\Zero}{\NO}{\Nat}{\Bin} & \triangleq & \Zero = \; \text{\coqe{bupa}}\; \NO \\[1em]
  \RS : \heteqP{\Succ}{\NS}{\Nat \rightarrow \Nat}{\Bin \rightarrow \Bin}
                              & \triangleq & \forall \ n \ m, \ n = \;
                                             \text{\coqe{bupa}}\; m \rightarrow
                                             \Succ \ n = \; \text{\coqe{bupa}}\;  (\NS
                                             \ m)  \\

  \end{array}
  $$
Additionally, one also needs to show that \coqe{N} satisfies 
an induction principle corresponding to the induction principle of \coqe{nat}.
Recall that the induction principle of \coqe{nat} has type
\begin{shaded}
\begin{coq}
  Definition nat_rect: forall P : nat -> Type, P O -> (forall n, P n -> P (S n)) -> forall n : nat, P n.
\end{coq}
\end{shaded}
Thus, the corresponding induction principle for \coqe{N} ought to have type
\begin{shaded}
\begin{coq}
  Definition nat_rect_N : forall P : N -> Type, P NO -> (forall n, P n -> P (NS n)) -> forall n : N, P n.
\end{coq}
\end{shaded}
Note that this induction principle is very different from
\coqe{N_rect}, the canonical induction principle derived for the
inductive definition of \coqe{N} (Figure~\ref{fig:nat-binnat}). 
Finally, we also need to prove that \coqe{nat_rect_N} is
parametrically related to \coqe{nat_rect} in the global context: 
$$
\gCtx_{\Nat}=\text{\coqe{(nat;N;R_nat_N), (O;NO;RO), (S;NS;RS)}}
$$

Using these definitions, it becomes possible to use parametricity to automatically convert a definition over \coqe{nat} that uses the induction principle \coqe{nat_rect} to an equivalent one over \coqe{N}. 
For instance, consider the definition of \coqe{plus} defined on \coqe{nat}.
By Corollary~\ref{white box param}, it is possible to automatically deduce that the function
\begin{shaded}
\begin{coq}
  Definition plus_nat_N (n m : N) : N := nat_rect_N (fun _ => N) m (fun _ res => NS res) n.
\end{coq}
\end{shaded}
\noindent
is parametrically related to \coqe{plus}, and thus behaves in the same
way, because \coqe{plus_nat_N} is equal to
$\primeTransExt{\mathtt{plus}}{\gCtx_{\Nat}}$.
This means that for all \coqe{n m : nat} and \coqe{n' m' : N}, the following holds:
\begin{shaded}
\begin{coq}
     n SimParam n' : nat PapParam N -> m SimParam m'  : nat PapParam N -> plus n m SimParam plus_nat_N n' m'  : nat PapParam N.
\end{coq}
\end{shaded}
That is, using the heterogeneous version of the parametricity translation addresses the anticipation problem identified above, because we do not need to rely on a common interface between \coqe{nat} and \coqe{N} and define all functions over this interface, as would be required in CoqEAL~\cite{cohenAl:cpp2013}.

\paragraph*{The Limits of Parametricity} 

However, using the heterogeneous parametricity translation to obtain
automatic transport still does not scale to dependent types, because
of the computational problem of parametricity: parametrically-related 
functions behave the same \emph{propositionally} but not \emph{definitionally}.

Let us go back to the \coqe{diff} example of Section~\ref{sec:param-lift}.
Using the white box FP to get a
parametrically-related definition of \coqe{diff} over \coqe{N}, we
could expect to get
\begin{shaded}
\begin{coq}
Fail Definition diff_N n (e : NO = NS n) : False :=
  let P_N := nat_rect_N (fun _ => Type) (NO = NO) (fun n _ => False) in
  eq_rect N NO (fun n' _ => P_N n') (eq_refl NO) (NS n) e.
\end{coq}
\end{shaded}
But this term does not typecheck, as the Coq error message explains:

\begin{shaded}
\begin{center}
\verb|The term "..." has type "|\coqe{P_N (NS n)}\verb|" while it is expected to have type "|\coqe{False}\verb|".|
\end{center}
\end{shaded}

This is because even though \coqe{nat_rect} and \coqe{nat_rect_N} are
parametrically related, they are not equal by conversion.
And indeed, the equality
\coqe{nat_rect_N _ PO PS (NS n) = PS n}
only holds propositionally, but not definitionally.
This means that the premise of Corollary~\ref{white box param} does not
hold here: the term \coqe{diff} does not
typecheck in the context $\gCtx_\Nat$ where each constant, and in
particular the constant \coqe{nat_rect}, is considered as a black box
and does not come with associated computational rules.
So, while moving to a heterogeneous presentation seems promising with
respect to the anticipation problem, it is insufficient to deal with
the computation problem of parametricity.

\section{Univalence is Not Enough}
\label{sec:univ-not-enough}

We now briefly review the notion of type equivalences (\S\ref{sec:typequiv}) and the univalence principle (\S\ref{sec:univalence}), and explain why univalence alone is not sufficient for automatic transport across equivalences (\S\ref{sec:univ-autolift}).

\subsection{Type Equivalences}
\label{sec:typequiv}
A function $f : A -> B$ is an {\em equivalence} iff there exists a function $g : B -> A$ together with proofs that $f$ and $g$ are inverse of each other. More precisely, the {\em section} property states that $\forall a:A, g(f(a))=a$, and the {\em retraction} property dually states that $\forall b:B, f(g(b))=b$. An additional condition between the section and the retraction, called here the {\em adjunction condition}, expresses that the equivalence is uniquely determined by the function $f$ (and hence that being an equivalence is proof irrelevant).

\begin{definition}[Type equivalence]
Two types $A$ and $B$ are equivalent, noted $\Eq{A}{B}$, iff there exists a function $f: A -> B$ that is an equivalence. 
\end{definition}

A type equivalence therefore consists of two {\em transport functions} (\ie~$f$ and $g$), as well as three properties. The transport functions are obviously computationally relevant, because they actually construct terms of one type based on terms of the other type. Note that from a computational point of view, there might be different ways to witness the equivalence between two types, which would yield different transports. 

Armed with a type equivalence $\Eq{A}{B}$, one can therefore {\em manually} port a library that uses $A$ to a library that uses $B$, by using the $A->B$
function in covariant positions and the $B->A$ function in contravariant
positions. 
However, with
type dependencies, all uses of transport at the term level can leak at
the type level. This leakage requires not only the use of sections or
retractions to deal with type mismatches, but also additional
properties relating existing functions, as illustrated in
\S\ref{sec:intro} with the fact that the equivalence is a homomorphism with respect
to the addition on natural numbers.

This also means that while the properties of an equivalence are not used computationally for transporting from $A$ to $B$ or vice versa, their computational content can matter when one wants to exploit the equivalence of constructors that are indexed by $A$ or by $B$. For instance, to establish that a term of type $T\ (g(f(a))$ actually has type $T\ a$, one needs to rewrite the term using the section of the equivalence---which means applying it as a (computationally-relevant) function.

\subsection{Univalence}
\label{sec:univalence}

Univalence is a principle that aligns type equivalence with propositional equality~\cite{voevodsky:cmu2010}.
\begin{definition}[Univalence]
For any two types $A$, $B$, the canonical map $(A=B) -> (\Eq{A}{B})$ is an equivalence.
\end{definition}
In particular, this means that $(A=B) \Eq (\Eq{A}{B})$. 
Therefore, univalence allows us to generalize Leibniz's principle of indiscernibility of identicals, to what we call the principle of {\em Indiscernibility of Equivalents}.

\begin{theorem}[Indiscernibility of Equivalents]
\label{th:ioe}
For any $P : \Type -> \Type$, and any two types $A$ and $B$ such that $\Eq{A}{B}$, we have $\Eq{P\ A}{P\ B}$.
\end{theorem}
\begin{proof}
Direct using univalence: $\Eq{A}{B} \implies A=B \implies P\ A = P\ B \implies \Eq{P\ A}{P\ B}$
\end{proof}

In particular, univalence promises immediate transport: if $A$ and $B$
are equivalent, then we can always convert some $P~A$ to some
(equivalent) $P~B$, for every inhabitant of $P~A$, even axioms.
\begin{corollary}[Black Box Fundamental Property]
\label{th:bbfp-intro}
For any $P : \Type -> \Type$, and any two types $A$ and $B$ such that $\Eq{A}{B}$, there exists a function $\uparrow_{\blacksquare}\  : P\ A -> P\ B$.
\end{corollary}
We call this result the  ``Black Box'' Fundamental Property because it can be used to blindly transport a term of type $P~A$ to a term of type $P~B$, without looking at its particular
syntactical structure. As such, it is very useful to solve the computational issue of
parametricity.

\paragraph{Realizing Univalence}

In CIC and MLTT, univalence cannot be proven and is therefore defined as an {\em axiom}. 
Because the proof of Theorem~\ref{th:ioe} starts by using the univalence axiom to replace type equivalence with propositional equality, before proceeding trivially with rewriting, it has no computational content, and hence we cannot exploit (axiomatic) univalence to reap the benefits of automatic transport of programs and their properties across equivalent types. 
It is important for transport to be {\em effective}, \ie~that it has computational content. 

Intuitively, an effective function ensures {\em canonicity}: it never
gets stuck due to the use of an axiom. Conversely, a function that uses an axiom and hence ``does not compute'' is called {\em ineffective}. By extension, a type equivalence $\Eq{A}{B}$ consisting of two functions $f: A->B$ and $g: B->A$ is said to be {\em effective} iff both $f$ and $g$ are effective functions.

To solve the issue of effectiveness, Cubical Type Theory has recently been
proposed~\cite{cubicaltt,vezzosiAl:icfp2019}. This theory is an extension of
MLTT in which n-dimensional cubes can be directly
manipulated, making it possible to define a notion of equality
between two terms as the type of the line (1-dimensional cube) between
those two terms. This way, the induced notion of equality is more
extensional than the usual Martin-Löf identity type, and it satisfies
univalence computationally, so the induced transports are effective.

\subsection{Univalence vs. Automatic Lifting}
\label{sec:univ-autolift}

However, even when it is effective, univalence {\em alone} is not enough
to support the automatic transport of functions that are defined on equivalent
types. 

Let us go back to the example of addition on natural
numbers. There exists a complicated but efficient definition of addition on  binary natural numbers, \coqe{plus_N}:\footnote{This function definition corresponds to the infix notation \coqe{+_N} used in \S\ref{sec:intro}.}
\begin{shaded}
\begin{coq}
  Definition plus_N (n m : N) : N := (* complex definition *).
\end{coq}
\end{shaded}
Showing most properties of \coqe{plus_N}, such as associativity and commutativity, is much more involved than their counterparts on \coqe{nat}. Ideally, after proving once and for all that \coqe{plus_N} is ``equal'' to \coqe{plus}, one would like to be able to obtain these theorems for free by transporting the proofs for \coqe{plus} on \coqe{nat}, \ie~rewriting through this ``equality''.

The problem is that even in a univalent type theory,  \coqe{plus_N} and
\coqe{plus} cannot be proven equal directly, because they are not defined on
the same type.
Indeed, the ``equality'' between \coqe{plus_N} and \coqe{plus} is
{\em heterogeneous} and only makes sense because there is an equivalence
between \coqe{nat} and \coqe{N}. 
This means that technically, the actual equality \coqe{e_nat_N} that
can be stated and proven is between the pairs \coqe{(nat; plus)} and
\coqe{(N; plus_N)} at the telescope type
\coqe{Sigma (A : Type), (A -> A -> A)}.
Then to transport the proof of commutativity of \coqe{plus}
\begin{shaded}
\begin{coq}
  Definition plus_comm : forall (n m : nat), plus m n = plus n m.
\end{coq}
\end{shaded}
 to a proof of commutativity of \coqe{plus_N}, one needs to exhibit the
predicate
\begin{shaded}
\begin{coq}
  P_comm := fun X _ => forall (n m : X.1), X.2 m n = X.2 n m
\end{coq}
\end{shaded}
to be passed to the eliminator of equality in order to define
\coqe{plus_N_comm} as
\begin{shaded}
\begin{coq}
  Definition plus_N_comm : forall (n m:N), plus_N m n = plus_N n m :=
     eq_rect (Sigma (A : Type), (A -> A -> A)) (nat; plus) P_comm plus_comm (N; plus_N) e_nat_N
\end{coq}
\end{shaded}
This generalization step, which can quickly become complex, cannot in general be automatically inferred, and so needs to be explicitly provided by the user.

\newcommand{\addp}{\AgdaData{addp}}

In Cubical Type Theory \cite{vezzosiAl:icfp2019}, one would rather rely on the
primitive notion of dependent path and transport the proof as depicted in Figure~\ref{fig:agdaex}.
\begin{figure}
  \includegraphics[width=0.8\textwidth]{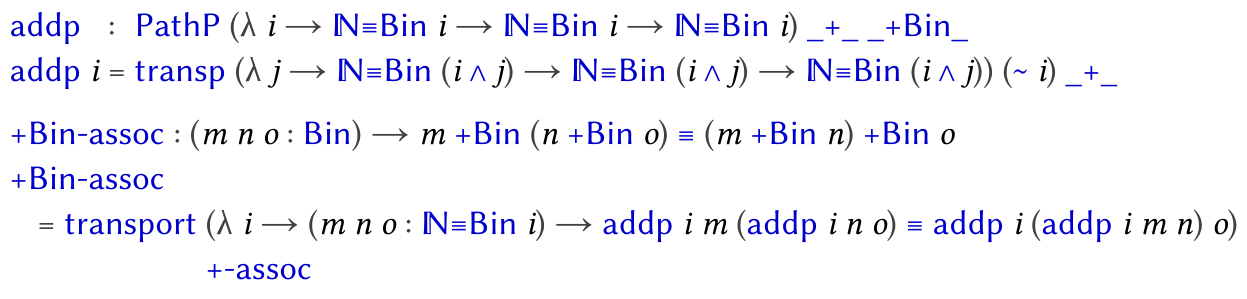}
  \caption{Transporting associativity from unary to binary naturals, in Cubical Agda (from \cite{vezzosiAl:icfp2019})}
  \label{fig:agdaex}
\end{figure}
In that case, one must still explicitly specify an abstraction of the
lemma statement to give to the transport function, and produce a
dependent path between the two notions of addition. In this example,
the addition on binary numbers $\AgdaData{\_+Bin\_}$ is defined
by a ``naive'' transport of the addition on unary numbers.  The direct
proof $\addp$ shows that it is related to addition
of unary numbers through the type equivalence
$\AgdaNat \AgdaData{$\equiv$} \AgdaData{Bin}$. Note that the proof of this equality for
the efficient addition on binary numbers would be much more involved,
but this is not the point here. The point we want to stress is that
transporting the proof term $\AgdaData{+-assoc}$  to the proof term
$\AgdaData{+Bin-assoc}$ requires the user to exhibit the predicate
$\lambda i -> (m ~ n ~ o : \AgdaNat \AgdaData{$\equiv$} \AgdaData{Bin} ~ i) ->
\addp~i~m~(\addp ~ i~n~o) \equiv \addp~i ~(\addp~i~m~n)~o$. 
We believe that, even in this simple example, a user would arguably appreciate 
some help from an automatic tool. 

This issue is very similar to the anticipation problem of parametricity
described in Section~\ref{sec:param-lift}. 
Indeed, the technique of using telescopes or dependent path types to encode heterogeneity
is akin to finding the right interface for the algebraic structure of a type.
But again, this does not scale to automation, and this limitation is independent of whether we are
in a univalent type theory or not.

However using univalence {\em does} solve the computation issue of
parametricity, as the function \coqe{diff} can be transported as a black
box, by simply using an equality \coqe{e'_nat_N} between
\coqe{(nat; (0,S))} and \mbox{\coqe{(N; (N0,NS))}} at type
\coqe{Pack := Sigma A : Type, A * (A -> A)}.
\begin{shaded}
\begin{coq}
Definition diff_N : forall (n:N) (e : N0 = NS n), False :=
  eq_rect Pack (nat; (0,S)) (fun X _ => forall (n:X.1) (e : fst X.2 = snd X.2 n), False)
                diff (N; (N0,NS)) e'_nat_N.
\end{coq}
\end{shaded}
Therefore, it seems that a combination of parametricity and univalence could address all the issues identified thus far.

\section{Univalent Parametricity in Action}
\label{sec:univ-param-acti}

This article develops the notion of {\em univalent parametricity} as a
fruitful marriage of univalence and parametricity, which leverages
their strengths while overcoming their limitations when taken in
isolation.
Specifically, univalent parametricity solves the anticipation problem
of parametricity by using (a variant of) the heterogeneous
parametricity translation, and solves the computation problem of
parametricity by using univalence.

Given two equivalent types, univalent parametricity can be used to automatically transport properties defined on one type---\eg~an easy-to-reason-about representation such as \coqe{nat}---to their counterparts on the other type---\eg~a computationally-efficient representation such as \coqe{N}.
Univalent parametricity provides the best of both parametricity and
univalence in that, to transport a term, we can use either the White Box FP 
(Th~\ref{th:wbfp-intro}) or the
Black Box FP (Th~\ref{th:bbfp-intro}), depending on the situation. 
In fact, the interplay between both modes is subtle in that white-box transport can automatically build univalent relations between two arbitrarily complex types based on user-provided relations, thereby inducing an equivalence that provides black-box transport for their inhabitants.

Univalent parametricity is a variant of the heterogeneous
parametricity translation $\heteqP{a}{b}{A}{B}$ introduced in
\S\ref{sec:param-not-enough}, simply noted~$\heteq{a}{b}{A}{B}$.
When $\heteq{a}{b}{A}{B}$ is inhabited, we say that $a$ and $b$ are
\emph{univalently related}. We sometimes omit $A$ and $B$ when they are clear from context.

The full development of univalent parametricity is in the following sections. In this section, we briefly illustrate univalent parametricity in action with  programs and proofs over \coqe{nat} and \coqe{N}. 
First, we illustrate transport {\`a} la carte, \ie~the possibility to refine automatic transport by establishing additional univalent relations (\S\ref{sec:talc}). Second, we show that univalent parametricity allows us to transport properties proven on \coqe{nat} to properties proven on \coqe{N} automatically (\S\ref{sec:auto-lifting}), and vice versa (\S\ref{sec:auto-comp}).

\subsection{Automatic Transport {\`a} la Carte}
\label{sec:talc}

Having proven the type equivalence between \coqe{nat} and \coqe{N}, we
can prove that they are univalently related,
\ie~$\heteq{\coqe{nat}}{\coqe{N}}{\Type}{\Type}$. Doing so induces an
automatic transport function, corresponding to the Black Box FP of
univalence (Th~\ref{th:bbfp-intro}). For instance we can transport a \coqe{square} function on
\coqe{nat} to an equivalent function on \coqe{N}:\footnote{ In the
  following, arithmetic operations in expressions are denoted with the
  same infix symbols (such as \coqe{+} and \coqe{*}); the actual
  operation is unambiguously determined by the type of its operands.
}
\begin{shaded}
\begin{coq}
Definition square (x : nat) : nat := x * x.
Definition square_N_bbox : N -> N := upaBB square.
\end{coq}
\end{shaded}

Note that from this section on, in code examples we use a general black-box transport operator \coqe{upaBB} whose source and target types are inferred from context, and whose  underlying equivalence is computed automatically, as will be explained in \S\ref{sec:coq}.

While \coqe{square_N_bbox} is an effective function that can be used to compute the square of any binary natural number, it is inherently inefficient computationally, because of the black-box nature of the transport: 
when applied, \coqe{square_N_bbox} first converts its \coqe{N} argument to an equivalent \coqe{nat}, applies the (slow) multiplication operation on \coqe{nat}, and finally converts back the \coqe{nat} result to a \coqe{N}:
\begin{shaded}
\begin{coq}
Check eq_refl : square_N_bbox = (fun x:N => upaBB (square (upaBB x))).
\end{coq}
\end{shaded}

At the cost of an additional proof effort, it is possible to establish that \coqe{mult} and \coqe{mult_N}
are univalently related, 
$\heteq{\text{\coqe{mult}}}{\text{\coqe{mult_N}}}{\text{\coqe{nat -> nat}}}{\text{\coqe{N -> N}}}$ 
(and likewise for \coqe{plus} and \coqe{plus_N}):
\begin{shaded}
\begin{coq}
Definition univrel_mult : mult ≈ mult_N.
\end{coq}
\end{shaded}

A first pay-off for this additional proof effort is that transport can now
automatically exploit such a relation, so that we can transport
\coqe{square} using the White Box FP of univalent parametricity, \ie~rewriting its body and exploiting the univalent relation between \coqe{mult} and \coqe{mult_N}:
\begin{shaded}
\begin{coq}
Definition square_N_wbox : N -> N := upaWB square.
\end{coq}
\end{shaded}
The notation \coqe{upaWB} corresponds to the
$\primeTransExt{\cdot}{\gCtx}$ notation of the White Box FP (Th~\ref{th:wbfp-intro}), 
where the global context $\gCtx$ is implicit. The management of the global
context and the definition of this (ad hoc) polymorphic operator are
realized in \Coq through typeclasses, as will also be explained in
\S\ref{sec:coq}.

The transported function \coqe{square_N_wbox} now computes directly on \coqe{N}, instead of converting back and forth and using the (slow) multiplication operation on \coqe{nat}.
\begin{shaded}
\begin{coq}
Check eq_refl : square_N_wbox = (fun x => (x * x)
\end{coq}
\end{shaded}

\subsection{Automatic Transport of Properties}
\label{sec:auto-lifting}

Establishing that two terms like \coqe{plus} and \coqe{plus_N} are univalently related is not only valuable from a computational point of view. It also enables the automatic transport of properties that involve such terms.

Without presenting the details of univalent parametricity yet, suffice it to say that the type \coqe{plus ≈ plus_N} actually unfolds to
\begin{shaded}
\begin{coq}
  forall (x : nat) (y : N), x = upaBB y -> forall (x' : nat) (y' : N), x' = upaBB y' -> x + x' = upaBB (y + y')
\end{coq}
\end{shaded}
\noindent which gives an extensional interpretation of the heterogeneous
equality between \coqe{plus} and \coqe{plus_N}, using univalent transport \coqe{upaBB} on terms of type \coqe{N}.

Then, thanks to the univalent relation \coqe{plus ≈ plus_N}, it is possible to automatically infer the type equivalence between the type \coqe{forall n m : nat, n + m = m + n}
and the type \coqe{forall n m : N, n + m = m + n}. Consequently, the {\em proof} of commutativity for \coqe{plus} can automatically be transported to a proof
of commutativity for \coqe{plus_N}:
\begin{shaded}
\begin{coq}
Definition plus_N_comm : forall n m : N, n + m = m + n := upaBB plus_comm.
\end{coq}
\end{shaded}

Note that here, we do not face the computation issue encountered by using only
parametricity (\S\ref{sec:param-lift}) because the
term \coqe{plus_comm} is transported as a black box, \ie without
recursively diving into its syntax.

In the same way, we can define the power function on both \coqe{nat}
and \coqe{N} and show that they are univalently related.
Then, the following very simple proof of an additive property of the
power function:
\begin{shaded}
\begin{coq}
Definition pow_prop : forall n:nat, 3 ^ (n + 1)  = 3 * 3 ^ n.
  intro n; rewrite plus_comm; reflexivity. 
Qed.
\end{coq}
\end{shaded}
\noindent
can be transported automatically to the power function on binary natural numbers:
\begin{shaded}
\begin{coq}
Definition pow_N_prop : forall n:N, 3 ^ (n + 1) = 3 * 3 ^ n := upaBB pow_prop.
\end{coq}
\end{shaded}
In contrast, because adding $1$ to a binary natural number is not an
operation that preserves the inductive structure of that number, a direct
proof of this lemma by induction on the binary natural number is much
more involved. 

\subsection{Automatically Computing in the Equivalent Representation}
\label{sec:auto-comp}

Univalent parametricity can also be used the other way around to prove
properties by computation on a type representation that is not always
effective.
Consider for instance the definition of a polynomial on natural
numbers
\begin{shaded}
\begin{coq}
Definition poly : nat -> nat := fun n => 12 * n + 51 * n ^ 4 - n ^ 5.
\end{coq}
\end{shaded}
\noindent
and consider proving that \coqe{poly 50} is bigger than some given
value, say \coqe{1000}.\footnote{We thank Assia
  Mahboubi for suggesting this example, taken from an actual mechanized mathematics exercise.} 
One would like to prove this by computation, \ie~by actually calculating
the value of \coqe{poly 50} and then simply comparing the result with \coqe{1000}. However, because the unary representation is very inefficient, evaluating \coqe{poly} at \coqe{50} already exceeds the stack capacity of the Coq runtime.
\begin{shaded}
\begin{coq}
Eval compute in poly 50.
Error Stack overflow
\end{coq}
\end{shaded}
Therefore, the proof that \coqe{poly 50} is bigger than \coqe{1000}
cannot be done by computation. Univalent parametricity can overcome this issue by transporting the inequality to be proven to an equivalent
one that uses the binary number representation. 
\begin{shaded}
\begin{coq}
Goal poly 50 >= 1000.
  replace_goal; now compute.
Defined.
\end{coq}
\end{shaded}
The tactic \coqe{replace_goal} automatically infers, from
\coqe{poly 50 >= 1000}, the univalently-related proposition on binary
natural numbers using the White Box FP.
Once the goal has been transported to a property on binary natural
numbers, it is possible to proceed by computation, which produces a
goal that can be solved automatically.

Note that automatic transport also works if we consider a slightly
more complex example, where polynomials are encoded as a list of
natural numbers representing its coefficients, together with a recursive
evaluation function \coqe{evalPoly}:

\begin{shaded}
\begin{coq}
Definition polyType := list nat.

Fixpoint evalPoly (p : polyType) (n : nat) (degree : nat) : nat :=
  match p with 
  | [] => 0
  | coef :: p => coef * n ^ degree + evalPoly p n (S degree) 
  end.

Infix "@@" := (fun p n => evalPoly p n 0).  
\end{coq}
\end{shaded}

Then, defining the polynomial \coqe{poly} in this setting and
evaluating it at \coqe{50} leads to the exact same issue. 

\begin{shaded}
\begin{coq}
Definition poly' : polyType := [0;12;0;0;51;1].

Eval compute in poly' @@ 50.
Error Stack overflow
\end{coq}
\end{shaded}
  
And the univalent parametricity framework allows again to transfer the
goal to binary numbers in order to solve it by computation.  
\begin{shaded}
\begin{coq}
Goal poly' @@ 50 >= 1000.
  replace_goal; now compute. 
Defined.
\end{coq}
\end{shaded}

Because univalent parametricity is defined on all of \CIC, this proof
technique also scales to the definitions of fixpoints. For instance, consider the following sequence definition:
\begin{shaded}
\begin{coq}
Fixpoint sequence (acc n : nat) :=
  match n with
    0 => acc
  | 1 => 2 * acc
  | 2 => 3 * acc
  | S n => (sequence acc n) ^ acc
  end.
\end{coq}
\end{shaded}
Indeed, one can generically show that fixpoints preserve univalently-related arguments, which means that sequences producing unary natural numbers can be transported automatically to equivalent sequences producing binary natural numbers.
\begin{shaded}
\begin{coq}
Definition sequence_prop : sequence 2 5 >= 1000.
  replace_goal; now compute. 
Defined.
\end{coq}
\end{shaded}

In summary, univalent parametricity follows the structural, white-box approach of parametricity to infer new univalent relations from existing ones, and can then exploit the induced equivalences as computational black boxes, as in univalence, to transport proofs and terms. The following sections develop the theory of univalent parametricity for \CC (\S\ref{sec:univparam}) and CIC (\S\ref{sec:inductives}), and its realization in the \Coq proof assistant (\S\ref{sec:coq}). Then, we revisit the examples of this section, explaining how each step is implemented (\S\ref{sec:extended-example}), and discuss a case study for integrating native datatypes in \Coq (\S\ref{sec:ffi}).

\section{Univalent Parametricity}
\label{sec:univparam}

We now turn to the formal development of univalent parametricity. In this section, we focus on \CC---extension to CIC is in \S\ref{sec:inductives}. 
We consider a type theory with the minimum requirements, namely the
Calculus of Constructions with universes and the univalence
axiom. 

We first discuss in Section~\ref{sec:univ_type} the proper way to
extend the parametricity translation in the universe to take
equivalences into account. We then provide the complete univalent
translation for \CC and present the Abstraction Theorem
and its proof (\S\ref{sec:translation}). 
Note that the development here is largely independent of any
particular realization of univalence, whether axiomatic or
computational. In our axiomatic setting, let us insist that sometimes, 
a direct use of the univalence axiom
would trivialize a proof but it would also suppress its
computational content, thus leading to a useless framework in
practice. In a setting where univalence is fully realized, such as in
Cubical Type Theory, those simpler proofs could be used (see also \S\ref{sec:related}).

In \S~\ref{sec:white-box-black}, we extend the translation by
taking a global context as input, which allows us to formulate the
\emph{White Box Fundamental Property} for \CC (Corollary~\ref{white
  box}).
This property can relate terms of completely different types,
such as inductively-defined and binary-encoded naturals, providing the
global environment provides witnesses that they are in univalent
relation. This is important because we want to be able to let
programmers define their own equivalences and thus get univalent
transport \emph{\`a la carte} (\S\ref{sec:up-constants}).
Univalent parametricity on types also entails type equivalence,
which allows us to  also state a \emph{Black Box Fundamental Property}
(Proposition~\ref{black box}), which says that when two types $A$ and
$B$ are in univalent relation, any \emph{open} term $t:A$ is in
univalent relation with its (black-box) transport of type $B$.

\subsection{Univalent Parametricity on the Universe}
\label{sec:univ_type}

To strengthen parametricity to deal with equivalences, the univalent parametricity translation 
$\uparamT{-}$ must strengthen the parametricity translation on the universe
$\Type_i$.
Several intuitive solutions come to mind, which however are not
satisfactory.

First, we could simply replace the relation demanded by parametricity
to be type equivalence itself,
\ie~$\uparamT{\Type_i} \ A \ B \defeq \Eq{A}{B}$. However, by doing
so, the abstraction theorem fails on $\vdash \Type_i :
\Type_{i+1}$. We would need to establish the fixpoint on the universe,
\ie~$\uparamT{\Type_i} : \uparamT{\Type_{i+1}} \ \Type_i \ \Type_i$,
but we have
$$
\uparamT{\Type_i} : \Type_i -> \Type_i -> \Type_{i+1} \neq
\Eq{\Type_i}{\Type_i}. 
$$
In words, on the left-hand side we have an arbitrary relation on $\Type_i$, while on the right-hand side, we have an equivalence.

Another intuitive approach is to state that $\uparamT{\Type_i} \ A \ B$
requires {\em both} a relation on $A$ and $B$ {\em and an
equivalence} between $A$ and $B$. While this goes in the right
direction, it is insufficient because there is no connection between
the two notions. This in particular implies that, when scaling up from
\CC to \CIC, the identity type---which defines the notion of
equality---will not satisfy the abstraction theorem of univalent
parametricity.
We need to additionally demand that the
relation {\em coincides with propositional equality} once the values
are at the same type.

Therefore, an inhabitant of $\uparamT{\Type_i}$
is given by a relation $R : A -> B -> \Type_i$ and an equivalence $e: \Eq{A}{B}$, 
together with a {\em coherence condition} between the relation and the 
equivalence.\footnote{
  Such an inhabitant is thus a dependent 3-tuple. We will later use syntactic sugar $t=(a;b;c)$ with accessors $t.1$ $t.2$ and $t.3$ for nested pairs to ease
  reading.}
This (crucial!) condition stipulates that the relation does coincide
with propositional equality up to a transport using the equivalence, \ie~for all $a:A$ and $b:B$, the following should hold:\footnote{From an equivalence $e: \Eq{A}{B}$ we can extract functions $\Transpe{B}{} : A -> B$ and $\Transpe{A}{} : B -> A$. We simply write $\Transpe{e}{}$ to refer to either one as required by the context; these correspond to univalent (aka. black-box) transport.
 Also, in the following, when $e$ is clear from the context, we often omit it and simply write $\Transp{~}$,
as in the \Coq examples of \S\ref{sec:univ-param-acti}.}
$$
\Eq{R~a~b}{(a={\Transpe{e}{b}})}
$$

This coherence condition allows us to show that once the relation is
fixed, the rest of the data is a mere proposition. That is, the fact
that there exists an equivalence satisfying the coherence condition
just characterizes the kind of relations that can be used, but it does
not provide additional structure. This way, univalent parametricity is
really just a restriction of parametricity to relations that
correspond to equivalences.

Therefore, for $\Type_i$, we want the translation to be:
\begin{align*}
\uparamT{\Type_i}\ A\ B  & \defeq \Sigma (R : A -> B -> \Type_i)
(e:\Eq{A}{B}).  \ \Pi a\  b. \Eq{(R \  a \  b)}{(a = \Transpe{e}{b})}
\end{align*}
That is, the translation of a type (when seen as a term) needs to
include a relation plus the fact that there is an equivalence, and
that the relation is coherent with equality.

We therefore need to distinguish between the translation
of a type $T$ occurring in a {\em term position} (\ie~left of the ``$:$''), 
translated as $\uparam{T}$ and the translation of a type $T$ occurring in a 
{\em type position} (\ie~right of the ``$:$''), translated as $\uparamT{T}$.\footnote{The possibility to distinguish the translation of a type on the left
and right-hand side of a judgment has already been noticed for
other translations that add extra information to types
by \citet{boulier:hal-01445835}.
For instance, to prove the independence of univalence with \CIC, they
use a translation that associates a Boolean to any type, \eg~$\trans{\Type_i} = (\Type_i \times \mathbb{B}, \textsf{true})$. Then a type on the left-hand side is translated as a 2-tuple and $\transType{A} = \trans{A}.1$. This
possibility to add additional information in the translation of a type
comes from the fact that types in \CIC can only be ``observed''
through inhabitance, that is, in a type position; 
therefore, the translation in term positions may collect
additional information.}
The abstraction theorem on $\Type_i$ enforces the definition of the
translation on the universe hierarchy to satisfy:\footnote{$\equiv$ denotes equality by conversion.} 
\begin{align*}
  \uparam{\Type_i} : & \ \uparamT{\Type_{i+1}} \ \Type_i \ \Type_i\qquad
                     \equiv \\
& \
\Sigma (R : \Type_i -> \Type_i -> \Type_{i+1})
  (e:\Eq{\Type_i}{\Type_i}).  \ \Pi a\  b. \Eq{(R \  a \  b)}{(a = \Transpe{e}{b})}.
\end{align*}
That is, $\uparam{\Type_i}$ must be itself a triple given by
\begin{align*}
    \uparam{\Type_i} & \defeq (\lambda \ (A \ B : \Type_i), \ \Sigma (R : A -> B -> \Type_i)
                                   (e:\Eq{A}{B}).  \\[0em]
                                  &  \qquad \Pi a b. \Eq{(R \  a \
                                    b)}{(a = \Transpe{e}{b})};
                                    \id{\Type_i} ; \univTerm{\Type_i})
\end{align*}
where $\id{\Type_i}$ is simply the identity equivalence on the
universe and $\univTerm{\Type_i}$ is a proof that the univalent
relation in the universe is coherent with equality on the universe as
given below. 
\begin{proposition}\label{prop:FP_Type}
  There exists a term
  $
  \univTerm{\Type_i} : 
  \Pi A \  B.\  
  \Eq{\uparamT{\Type_i}~A~B}{(A=B)}.
  $ 
\end{proposition}

\begin{proof}

  The definition of $\univTerm{\Type_i}$ crucially relies on univalence (and actually is
  equivalent to it), so this equivalence is not effective.

  This result requires functional extensionality, \ie~the fact that the canonical map 
  $$
  f = g  -> \Pi (x:A). f \  x = g \  x
  $$
  is an equivalence. This property is in fact a consequence of
  univalence~\cite{hottbook}.

  By univalence and rearrangement of dependent sums, the type
  $\uparamT{\Type_i}~A~B$ is equivalent to
  $$
  \Sigma (e:\Eq{\Type_i}{\Type_i}) (R : \Type_i -> \Type_i ->
  \Type_{i+1}). \ \Pi a\  b. R \  a \  b = (a = \Transpe{e}{b})
  $$
  which by functional extensionality is equivalent to
  $$
  \Sigma (e:\Eq{\Type_i}{\Type_i}) (R : \Type_i -> \Type_i ->
  \Type_{i+1}). \ R = (\lambda~a~b.~a = \Transpe{e}{b})
  $$
  But $\Sigma (R : \Type_i -> \Type_i ->
  \Type_{i+1}). \ R = (\lambda~a~b.~a = \Transpe{e}{b})$ is a
  singleton type, which is always contractible. 
  So we get an equivalence
  $$
  \Eq{\uparamT{\Type_i}~A~B}{(\Eq{A}{B})}
  $$
  and we conclude by univalence again.
\end{proof}

\subsection{The Univalent Parametricity Translation}
\label{sec:translation}

We now turn to the full definition of the univalent parametricity
translation on the whole syntax of \CC, including variables,
application and lambda expressions, as a variation on the
parametricity translation in the style of Bernardy \etal (recall
Figure~\ref{fig:param} of \S\ref{sec:param-as-logic}).
Figure~\ref{fig:univ} shows how to extend the parametricity
translation to force the relation defined between two types to
correspond to a type equivalence with the coherence condition, as exposed in 
\S\ref{sec:univ_type}.
Note that the translation does not target \CC but rather \CICU, which
is \CIC augmented with the univalence axiom. We write $\Gamma \vdashu t:T$ to
stipulate that the term is typeable in \CICU. 

\begin{figure}[t]
\begin{align*}
\uparam{\Type_i} & \defeq (\lambda \ (A \ B : \Type_i), \ \Sigma (R : A -> B -> \Type_i)
                                   (e:\Eq{A}{B}).  \\[0em]
                                  &  \qquad \Pi a b. \Eq{(R \  a \  b)}{(a = \Transpe{e}{b})};  \id{\Type_i} ; \univTerm{\Type_i}) \\[0em]
\uparam{\Pi a:A. \ B}  & \defeq (\lambda \ (f : \Pi a:A. B) ~ (g: \Pi
                         a':A'. B'). \\ & \hspace{3.5em} \Pi
                              (a:A) (a':A') (\epsTrans{a}: \uparamT{A} \ a \
a'). \ \uparamT{B}~(f\ a)~(g\ a');  \\
& \qquad \equivPi~A~A'~ \uparam{A}~B~B'~ \uparam{B}; \univTerm{\Pi}~A~A'~ \uparam{A}~B~B'~ \uparam{B}) \\[0em]
\uparam{x} & \defeq \epsTrans{x} \\[0em]
\uparam{\lambda x:A. t} & \defeq \lambda (x:A) (x':A')
                          (\epsTrans{x}:\uparamT{A}\ x\ x'). \ \uparam{t}\\[0em]
\uparam{t\ u} & \defeq \uparam{t} \ u \ u' \ \uparam{u} \\[1em]
\uparamT{A} & \defeq \uparam{A}.1 \quad
\uparamEq{A} \defeq  \uparam{A}.2 \quad
\uparamCoh{A} \defeq  \uparam{A}.3 \\[1em]
\uparamT{\cdot} & \defeq \cdot \\
\uparamT{\Gamma,x:A} & \defeq \uparamT{\Gamma},
  x:A,\ x':A',\ \epsTrans{x} : \uparamT{A} \ x \ x'
\end{align*}
\caption{Univalent parametricity translation for \CC}
\label{fig:univ}
\end{figure}

As explained in the previous section, the definition of the translation
of a type $A$ is more complex than that of Figure~\ref{fig:param}
because in addition to the relation $\uparamT{A}$, we need an
equivalence $\uparamEq{A}$ and a witness $\uparamCoh{A}$ that the
relation is coherent with equality.
This is why the translation of dependent products makes use of two
additional terms $\equivPi$ and $\univTerm{\Pi}$ that will be
explained during the proof of the Abstraction Theorem.

For the other terms, the translation does not change with respect to
parametricity except that $\uparamT{-}$ must be used accordingly when
we are denoting the relation induced by the translation and not
the translation itself.

\begin{theorem}[Abstraction theorem]
\label{thm:abstraction}
  If $\Gamma \vdash t :A$ then 
  $\uparamT{\Gamma} \vdashu \uparam{t} : \uparamT{A}\ t\ t'$.
\end{theorem}

\begin{proof}
  The proof is a straightforward induction on the typing derivation.
  The only cases that differ from Theorem~\ref{thm:abstraction-param}
  are the typing rules for the universe and for the type of dependent functions.
  The case of the universe has already been addressed in \S~\ref{sec:univ_type}.
  For dependent products, we first need to provide a term that
  witnesses the fact that the dependent product is congruent with
  respect to equivalences in the following sense:
  \begin{align*}
    \equivPi : \ & \Pi \ (A \  B : \Type_i) \  (\UR{AB}: \uparamT{\Type_i} ~ A ~ B).
    \\  & \Pi \ (P : A -> \Type_i) \ (Q : B -> \Type_i) \
         \\ & \phantom{\Pi \ } (\UR{PQ}  :  \Pi (a:A) \  (b:B). \ \proji{\UR{AB}}~a~b
          ->  \uparamT{\Type_i} ~ (P a) ~ (Q b))).
    \\ & \phantom{\Pi \ } \Eq{(\Pi (a:A). \  P \  a)}{(\Pi (b:B).  \  Q \  b)}
  \end{align*}
  In particular, we have  
  $ \projii{\UR{AB}} : \Eq{A}{B}$ and $\lambda~a~b~r, \
\projii{(\UR{PQ}~a~b~r)} : \Pi (a:A) \  (b:B). \ \proji{\UR{AB}}~a~b -> \Eq{P \  a}{Q \  b}$.
Using the coherence condition $\projiii{\UR{AB}}$ between $ \proji{\UR{AB}}~a~b $ and $a = \Transp{b}$, this boils down to 
$\Pi (a:A). \  \Eq{P \  a}{Q \  (\Transp{a})}$.

At this point we can apply a standard result of HoTT, namely
\coqe{equiv_functor_forall} in the Coq HoTT
library~\cite{Bauer:2017:HLF:3018610.3018615}. This lemma requires functional extensionality in the proof that the two transport functions form an equivalence.\footnote{The definition of the inverse function requires using 
the retraction, and the proof that it forms a proper equivalence requires the adjunction condition (\S\ref{sec:typequiv}).
This means that the dependent function type would not be univalent if we
replaced type equivalence with a simpler notion, such as the possibility
to go from one type to another and back, or even by isomorphisms.}

The second step is to provide a proof that the univalent (pointwise)
relation on the dependent product is coherent with equality up to the
equivalence above; \ie~we need to define a term
\begin{align*}
\univTerm{\Pi} : \ & \Pi \ (A \  B : \Type_i) \  (\UR{AB}: \uparamT{\Type_i} ~ A ~ B).
  \\  & \Pi \ (P : A -> \Type_i) \ (Q : B -> \Type_i) \
 \\ & \phantom{\Pi \ } (\UR{PQ}  :  \Pi (a:A) \  (b:B). \ \proji{\UR{AB}}~a~b
      ->  \uparamT{\Type_i} ~ (P\ a) ~ (Q\ b))).
  \\ & \phantom{\Pi \ } \Pi f \  g.  \Eq{
       (\Pi (a:A) (b:B) (r:\proji{\UR{AB}}~a~b). \ \proji{(\UR{PQ}~a~b~r)}~(f\ a)~(g\ a'))
       }{(f={\Transp{g}})} 
\end{align*}
This part is quite involved. In essence, this is where we prove that transporting in many hard-to-predict places is equivalent to 
transporting only at the top level.
This is done by repeated use of commutativity lemmas of transport of
equality over functions.
We refer the reader to the Coq development for more details.
\end{proof}

\subsubsection{Prop}
\label{sec:prop}

The definition of the univalent parametricity translation of
Figure~\ref{fig:univ} does not deal with the universe $\Prop$ of
proposition, but it can be treated in the same way
as $\Type_i$ because $\Prop:\Type_i$ is a universe that also enjoys
the univalence axiom. The only specificity of $\Prop$ is its
impredicativity, which does not play any role here.

It is also possible to state a stronger axiom on $\Prop$
called \emph{propositional extensionality}, which uses logical
equivalences instead of type equivalences in its statement:
$$
\Eq{(P=Q)}{(P <=> Q)}.
$$
This axiom cannot be deduced from univalence alone,
one would need proof irrelevance for $\Prop$ as well.
As we are looking for the minimal amount of axioms needed for
establishing univalent parametricity, we do not make use of this stronger axiom.

Note that exploiting the fact that $\Prop$ is proof
irrelevant, $\uparamT{\Prop}~P~Q$ boils down to 
$$
\Sigma (R : P -> Q -> \Type). \ ( P <=> Q) * (\Pi (p:P) \ (q:Q) . \
        \mathtt{IsContr} (R~p~q)).
$$
where $*$ is the product of types and $\mathtt{IsContr} \ A$ says that
$A$ is contractible, \ie it has a unique inhabitant.
This is because for all $p$ and $q$, the type $(p={\Transp{q}})$ is
contractible, and being equivalent to a contractible type is the same
as being contractible.
The definition we obtain in this case coincides with the definition
of parametricity with uniformity of propositions developed by \citet{DBLP:journals/corr/AnandM17} (more details in
\S\ref{sec:related}).

\subsection{White Box and Black Box Fundamental Properties}
\label{sec:white-box-black}

The extension of the parametricity translation that takes a global context into account
presented in order to handle heterogeneous instance (\S\ref{sec:het-param}) 
can also be applied to univalent parametricity.
Recall that we consider a global context $\gCtx$ to be the following telescope $\gCtx_n$ defined as:
\begin{align*}
\gCtx_0 & = \cdot \\
\gCtx_1 & = (\fstTrans{c}_1 : \fstTrans{A}_1 ;\ \sndTrans{c}_1 :\sndTrans{A}_1 ;\ \RTrans{c}_1 : 
\uparamWB{A_1}{\gCtx_0}\ \fstTrans{c}_1\ \sndTrans{c_1}) \\
\ldots & \\
\gCtx_n & = \gCtx_{n-1}, 
            (\fstTrans{c}_n : \fstTrans{A}_n ;\ \sndTrans{c}_n : \sndTrans{A}_n ;\ \RTrans{c}_n : \uparamWB{A_n}{\gCtx_{n-1}}\ \fstTrans{c}_n\ \sndTrans{c}_n)
\end{align*}
For simplicity, we reuse the notation of Section~\ref{sec:het-param}
to get a similar White Box FP for univalent
parametricity. In particular, the definition of
Figure~\ref{fig:univ} is likewise extended on constants as follows:
$$\uparamWB{\fstTrans{c}}{\gCtx} \defeq \RTrans{c} \text{ when }
(\fstTrans{c} :\_\ ;\ \sndTrans{c}:\_\ ;\ \RTrans{c}: \_) \in \gCtx$$

\begin{corollary}[White Box Fundamental Property]
  \label{white box}
  If $\contextFst{\gCtx}~\vdash a:A$ then
  $\contextSnd{\gCtx}~\vdash \primeTransExt{a}{\gCtx}: \primeTransExt{A}{\gCtx}$ and
  $\contextAll{\gCtx}~\vdashu \uparamWB{a}{\gCtx} : \uparamTWB{A}{\gCtx}\ a\ \primeTransExt{a}{\gCtx}$.
\end{corollary}

\begin{proof}
  Similar to the proof of Corollary~\ref{th:wbfp-intro}.
\end{proof}

As we have done for parametricity, we introduce the notation
$$
\heteq{a}{b}{A}{B} \defeq \uparamTWB{A}{\gCtx}~a~b \qquad\text{ (when $B = \primeTransExt{A}{\gCtx}$)}
$$
which relates two terms $a$ and $b$ at
two related---but potentially different---types $A$ and $B$.
This notation only makes sense when
$
B = \primeTransExt{A}{\gCtx}
$
in the current global context $\gCtx$, and we implicitly assume it is
the case when we use this notation.

The White Box FP is the usual result obtained using
parametricity, just rephrased using a context of global constants.
But in the case of univalent parametricity, we have an additional
property coming from the use of equivalences, which is {\em oblivious} to the
structure of the term, and is defined for any well-typed terms, even
in an open context, and in particular in presence of axioms.
Indeed, once we know that two types $A$ and $B$ are univalently
related, there is a canonical term in $B$ related to any term $t$ of
type $A$ obtained from the following property:

\begin{proposition}[Black Box Fundamental Property]
  \label{black box}
  Let $A$ and $B$ be two types such that $\heteq{A}{B}{\Type}{\Type}$,
  then for all $t:A$, there is a proof that $\heteq{t}{\Transp{t}}{A}{B}$.
\end{proposition}
\begin{proof}
  By reflexivity of equality, $\Transp{t} =_B \Transp{t}$, which gives
  the desired result by using the coherence condition between equality
  and the relation. 
\end{proof}

Note that the Black Box FP is internal to the
theory, as opposed to the White Box FP, which is external. In
particular, this means that the Black Box FP is
valid in any context on open terms, which is crucial to define
automatic transport.

\subsection{Univalent Parametricity for Transport {\`a} la Carte}
\label{sec:up-constants}

The Black Box Fundamental Property (Proposition~\ref{black box}) is a key advantage of univalent parametricity over traditional parametricity,
because knowing that two types $A$ and $B$ are univalently related is enough to get a transport function from $A$ to $B$, whereas traditional parametricity requires the exact definition of the term $a$ in $A$ in order to compute its counterpart in $B$.
Now, it remains to investigate how to determine that two types are
univalently related.

The White Box FP (Corollary~\ref{white box}) allows us to enrich
the univalent relation ``outside the diagonal'', \ie~to provide
heterogeneous instances in the global context.
For instance, we can relate unary naturals \coqe{nat} and binary naturals
\coqe{N}, \ie~$\heteq{\coqe{nat}}{\coqe{N}}{\Type}{\Type}$.
By combining this basic relation with the fact that type constructors
are univalently parametric, it is possible to automatically derive 
that:
$$
\heteq{\coqe{nat -> nat -> nat}}{\coqe{N -> N -> N}}{\Type}{\Type}
$$
But more interestingly, not only univalent relation instances between
types can be added, but also instances between any two terms, seen as new constants of the theory.
For instance, consider the case of the definitions of the addition on
unary and binary natural numbers \coqe{plus : nat -> nat -> nat}
and \coqe{plus_N : N -> N -> N}. One can show that they are
univalently related
$$
\heteq {\coqe{plus}}{\coqe{plus_N}}{\coqe{nat -> nat -> nat}}{\coqe{N -> N -> N}}
$$
which, as illustrated in \S\ref{sec:univ-param-acti}, allows automatic
transport to be more computationally efficient, and proofs of results
involving addition to be transported automatically.
Indeed, there exists a context $\gCtx$ that relates
$$
\contextFst{\gCtx} \equiv
\mbox{\coqe{nat:Type, plus:nat -> nat -> nat, mult:nat -> nat -> nat}}
$$
and
$$
\contextSnd{\gCtx} \equiv
\mbox{\coqe{N:Type, plus_N:N -> N -> N, mult_N:N -> N -> N}}.
$$
As discussed in Section~\ref{sec:talc}, applying the White Box FP to the term \coqe{square} directly gives us the term \coqe{square_N_wbox} together with a proof that it is related to \coqe{square} in the context $\gCtx$.

But the White Box FP also automatically gives us that the two following types are related
$$
\heteq{\coqe{forall (varn varm : nat), plus varn varm = plus varm varn}}{\coqe{forall (varn varm : N), plus_N varn varm = plus_N varm varn}}{\Type}{\Type}
$$
because
\coqe{forall (varn varm : N), plus_N varn varm = plus_N varm varn} is
the prime translation of
\coqe{forall (varn varm : nat), plus varn  varm = plus varm varn}.
Then, using the Black Box FP provides a direct transport
\coqe{plus_N_comm} of the proof of commutativity of addition
\coqe{plus_comm}.
Note that the computational content of \coqe{plus_comm} is not
used to derive \coqe{plus_N_comm}.
Therefore, univalent parametricity naturally provides a framework to
transport functions and proofs of theorems {\`a} la carte, depending
on the univalent relation context that has been specified by the user, thanks to the interplay between white-box and black-box transports.

\section{Univalent Parametricity for Inductive Types}
\label{sec:inductives}

We now turn to the extension of univalent parametricity in theories that provide inductive types, such as the Calculus of Inductive Constructions (\CIC)~\cite{Paulin15}. 
We first give the general idea of the approach, and then proceed step-by-step, first considering dependent pairs (\S\ref{sec:dependent pairs}), records (\S\ref{sec:records}), parameterized recursive inductive families
(\S\ref{sec:paramrecinductives}), and finally indexed inductive types
(\S\ref{sec:gadts}). 

An inductive type is defined as a new type constructor, together with
associated constructors and an elimination principle.\footnote{There
  is an equivalent presentation of inductive types with pattern
  matching instead of eliminators~\cite{Goguen2006}. In Coq,
  eliminators are automatically inferred and defined using pattern
  matching.}
For instance, the inductive type of lists is\footnote{In this section,
to ease the reading, we navigate between the syntax of \CIC and the
one of Coq when appropriate.}
\begin{shaded}
\begin{coq}
 Inductive list (A : Type) : Type :=
     nil : list A
  | cons : A -> list A -> list A
\end{coq}
\end{shaded}
where \coqe{nil} and \coqe{cons} are the constructors of the inductive
type. 
The associated eliminator is 
\begin{shaded}
\begin{coq}
list_rect : forall (A : Type) (P : list A -> Type), P nil ->  (forall (a : A) (l : list A), P l -> P (a :: l)) 
                     -> forall l : list A, P l.
\end{coq}
\end{shaded}

\newcommand\indI{\mathtt{I}}
Let us consider an inductive type $\indI:\typeOfI$, 
with constructors
$\mathtt{I\_c}_i:\typeOfci$ and elimination principle
$\mathtt{I\_rect}:\typeOfrect$.
Using the global context presentation of Section~\ref{sec:white-box-black}, the introduction of this new inductive type in our framework amounts to the ability to extend the current global context $\gCtx$ with the following triple for the type constructor
$$(\indI:\typeOfI; \indI:\typeOfI; \RTrans{\indI} : \uparamTWB{\typeOfI}{\gCtx}~\indI~\indI)$$
and similar triples for each constructor and for the eliminator.
We can then directly use the White Box FP (Corollary~\ref{white
  box}) extended with this inductive type. 
  To sum up, the addition of the triples for the inductive type $\indI$
amounts to giving the following terms\footnote{Here, and in the following, we use the notation $\heteq{a}{b}{A}{B}$ extensively to avoid explicitly mentioning the global context.}
$$
\left[
  \begin{array}{lcl}
     \RTrans{\indI} & :&
  \heteq{\indI}{\indI}{\typeOfI}{\typeOfI}
  \\
  \RTrans{\mathtt{I\_c}_i} & :&
  \heteq{\mathtt{I\_c}_i}{\mathtt{I\_c}_i}{\typeOfci}{\typeOfci}
  \\
  \RTrans{\mathtt{I\_rect}} & :&
  \heteq{\mathtt{I\_rect}}{\mathtt{I\_rect}}{\typeOfrect}{\typeOfrect}
  \end{array}
  \right.
$$
When the above terms exist, we say that the inductive type,
constructors and elimination principle are univalently parametric.

\subsection{Dependent Pairs}
\label{sec:dependent pairs}

In \CIC, dependent pairs are defined as the inductive family:
\begin{shaded}
\begin{coq}
Inductive sigT (A : Type) (B : A -> Type) : Type := 
        existT : forall x : A, B x -> sigT A B.
\end{coq}
\end{shaded}

\noindent Thus, the unique constructor of a dependent pair is \coqe{existT} and
the elimination principle is given by 
\begin{shaded}
\begin{coq}
sigT_rect : forall (A : Type) (P : A -> Type) (P_0 : sigT A P -> Type), 
       (forall (x : A) (p : P x), P_0 (x; p)) -> forall s : sigT A P, P_0 s
\end{coq}
\end{shaded}

As common, we use the notation $\Sigma a:A. \ B$ to denote \coqe{sigT A (fun a => B)}, similarly to dependent type theories where pair types are part of the syntax~\cite{MARTINLOF197573}.

\begin{proposition} \label{prop:FP_Sigma}
  There exists a term
  $$\Univ{\Sigma} :
  \heteq{\Sigma}{\Sigma}{\Pi (A:\Type_i).\ (A -> \Type_i) -> \Type_i}
  {\Pi (A:\Type_i).\ (A -> \Type_i) -> \Type_i}$$
\end{proposition}

\begin{proof}
  The main steps of the construction are similar to those for the
  dependent product. Unfolding the definitions, giving a term
  $\Univ{\Sigma}$ amounts to giving an inhabitant of $\uparamT{\Type_i}$
  given two terms $\UR{AB} : \heteq{A}{B}{\Type_i}{\Type_i}$ and
  $\UR{PQ} : \heteq{P}{Q}{A -> \Type_i}{B -> \Type_i}$.
  
  First, the univalent relation $R_\Sigma$ between $\Sigma a:A. \  P \  a$ and $\Sigma
  b:B.  \  Q \  b$ is defined as
  $$
  R_\Sigma \defeq \lambda (p: \Sigma a:A . \ P~a) (q:\Sigma b:B.\
  Q~b). \ \Sigma (\ur{pq} : \proji{\UR{AB}} ~ \proji{p} ~{\proji{q}}). \
  \proji{(\UR{PQ} ~ \proji{p} ~ \proji{q} ~ \ur{pq})} ~ {\projii{p}} ~
  {\projii{q}}.
  $$
  It naturally requires the first and second elements of the pair to
  be related at the corresponding types.

Second, the proof that $\Eq{\Sigma a:A. \ P \ a}{\Sigma b:B. \ Q \ b}$
also follows from a standard result of HoTT, namely
\coqe{equiv_functor_sigma} in the Coq HoTT library.
Contrarily to the dependent product, which requires functional extensionality, this lemma does not require any axiom.

Finally, the proof that the relation is coherent with equality is the novel part
required by univalent parametricity. 
This means that we need to define a term:
\begin{align*}
\univTerm{\Sigma} : \ & \Pi \ (A \  B : \Type_i) \  (\UR{AB}: \heteq{A}{B}{\Type_i}{\Type_i}).
  \\  & \Pi \ (P : A -> \Type_i) \ (Q : B -> \Type_i) \
        (\UR{PQ}  : \heteq{P}{Q}{A -> \Type_i}{B -> \Type_i}).
\\ & \phantom{\Pi \ } \Pi x \  y.  \Eq{(\heteq{x}{y}{\Sigma a:A. \  P \  a}{\Sigma b:B.  \  Q \  b})}{(x={\Transp{y}})}
\end{align*}
Instead of building the equivalence explicitly with the transport functions and their associated section and retraction proofs, this equivalence can be conveniently proven by composition of equivalences. 
Specifically, we rely on a decomposition of equality for dependent sums:
$$
( \heq{x}{y} ) \equiv 
( \Sigma p :\heq{x.1}{y.1}. \ \heq{x.2}{y.2} ) \Eqr 
( \Sigma p : x.1 = \Transp{y.1}. \ x.2 = \Transp{y.2} ) \Eqr 
( x={\Transp{y}} )
$$

Note that the last equivalence above
is the counterpart of functional extensionality for dependent
function types. 
The main difference is that this equivalence is effective as it can be
proven by elimination of dependent pairs.
\end{proof}

With $\Univ{\Sigma}$ defined, we can extend the global context for the
constructor and eliminator as well.
\begin{proposition} \label{prop:FP_Sigma_rest}
  There is a term
  $$\Univ{\text{\coqe{existT}}} :
  \heteq{\text{\coqe{existT}}}{\text{\coqe{existT}}}{{\mathtt{T_{ex}}}}{\mathtt{T_{ex}}}$$
  where
  $\mathtt{T_{ex}}  \defeq
  \Pi (A : \Type) (P : A -> \Type) (x: A). \ P ~ x -> \Sigma (x :
  A). \ P ~ x$
  and similarly for \coqe{sigT_rect}.
\end{proposition}
\begin{proof}
Direct by induction on the structure of a dependent pair type.
\end{proof}

\subsection{Record Types}
\label{sec:records}

The treatment of dependent pairs above scales to dependent records, by considering their encoding as iterated dependent pairs.

To illustrate, let us consider the example of a simple library record type \coqe{Lib}, which abstracts over an indexed container type constructor \coqe{C}, and packages functions \coqe{head} and \coqe{map} together with a property on their composition:
\begin{shaded}
\begin{coq}
Record Lib (C : Type -> nat -> Type) :=
  { head : forall {A : Type} {n : nat}, C A (S n) -> A;
    map : forall {A B} (f :A -> B) {n}, C A n -> C B n;
    prop : forall n A B (f : A -> B) (v : C A (S n)), head (map f v) = f (head v)}.
\end{coq}
\end{shaded}
Like all record types, \coqe{Lib} can be formulated in terms of nested dependent pairs. 
This means that, for any \coqe{C : Type -> nat -> Type}, 
\coqe{Lib C} is equivalent to
\begin{shaded}
\begin{coq}
  Lib' C := Sigma (hd : forall A n. C A (S n) -> A),
                   Sigma (map : forall A B (f:A -> B) n, C A n -> C B n),
                       forall n A B (f : A -> B) (v : C A (S n)), hd (map f v) = f (hd v).
\end{coq}
\end{shaded}

The fact that \coqe{Lib'} is univalently parametric directly follows from the
abstraction theorem of \CC extended with dependent pairs.
To conclude that \coqe{Lib} is univalently parametric, 
we use the fact that a type
family equivalent to a univalently parametric type family is itself
univalently parametric.

This approach to establish the univalent parametricity  record type via its encoding with dependent pairs can be extended to any record type. 
We have automatized this principle in our Coq framework as a tactic, by reusing an idea used in the HoTT library that allows automated inference of type equivalence for records with their nested pair types formulation. This tactic can be used to automatically prove that a given record type is univalently parametric (provided its fields are).

\subsection{Parameterized Recursive Inductive Families}
\label{sec:paramrecinductives}

To establish the univalent parametricity of a {\em parameterless} recursive inductive type, such as natural numbers with zero and successor, we can simply use the canonical structure over the identity equivalence, with equality as univalent relation and trivial coherence.
However, whenever an inductive type has parameters, the situation is more
complex.\footnote{In this work the distinction between 
{\em parameters} and {\em indices} for inductive types is important. 
A parameter is merely indicative that the type behaves
{\em uniformly} with respect to the supplied argument. For instance $A$ in 
 \mbox{\texttt{list} $A$} is a parameter. Thus the choice of $A$ only affects the type of elements inside
the list, not its shape. In particular, by knowing $A$ for a given list, we cannot infer which constructor was used to construct the list. 
On the other hand, $n$ in \mbox{\texttt{Vect} $A\; n$}  is an index. By knowing the value of an index, one can infer which constructor(s) may or may not have been used to create the value. For instance, a value of type \mbox{\texttt{Vect} $A\;0$} is necessarily the empty vector. We address indexed inductive types in \S\ref{sec:gadts}.}

Let us develop the case of lists and define the term
$$\Univ{\mathtt{list}} :
\heteq{\mathtt{list}}{\mathtt{list}}{\Type_i -> \Type_i}
{\Type_i -> \Type_i}.$$
Unfolding the definitions, giving a term
  $\Univ{\mathtt{list}}$ amounts to exhibiting an inhabitant of $\uparamT{\Type_i}$
  given two types $A$ and $B$ related by $\UR{AB} : \heteq{A}{B}{\Type_i}{\Type_i}$.

The univalent relation on lists is given directly by
parametricity. Indeed, following the work of \citet{bernardyAl:jfp2012} on the inductive-style translation, the inductive type
corresponding to the transport of a relation between \coqe{A} and
\coqe{B} to a relation between \coqe{list A} and \coqe{list B} is
given by:

\begin{shaded}
\begin{coq}
Inductive UR_list A B (R : A -> B -> Type) : list A -> list B -> Type :=
  UR_list_nil : UR_list R nil nil
| UR_list_cons : forall a b l l',  (R a b) -> (UR_list R l l') ->  UR_list R (a::l) (b::l').
\end{coq}
\end{shaded}
 This definition captures the fact that two lists are related iff they are of the same length and pointwise-related.
Then, the univalent relation is given by 
$$
R_{\mathtt{list}} \defeq \lambda (l: \mathtt{list} \ A) (l' :
  \mathtt{list} \ B) . \ \mathtt{UR\_list} \ A \ B \ \proji{\UR{AB}} \ l \ l'
$$

Then, we need to show that 
$$
\Eq{\texttt{list}\ A}{\texttt{list}\ B}
$$
knowing that $\Eq{A}{B} $ from $\UR{AB}$.
The two transport functions of the equivalence $\Eq{\texttt{list}\ A}{\texttt{list}\ B}$ can be defined by induction on the structure of the list (\ie using the eliminator
\coqe{list_rect}). They both simply correspond to the usual map operation on lists. The proofs of the section and retraction are also direct by induction on the structure of the list, and transporting along the section and retraction of $\Eq{A}{B}$.

Similarly to dependent pairs, the proof that the relation is coherent with
equality relies on the following decomposition of equality between lists:
$$ 
\Pi \ A \ B \  (e: \Eq{A}{B}) \ l \ l' . \ \Eq{(\mathtt{UR\_list}\ A \ B  \ (\lambda
\ a \ b . a = \ \Transp{b}) \ l \ l' ) } {(l = \ \Transp{ l'})}
$$

Indeed, using this lemma, the coherence of the univalent relation with
equality is easy to infer:
$$
( \heq{l}{l'} ) \equiv
(\mathtt{UR\_list} \ A \ B \ \proji{\UR{AB}} \ l \ l') \Eqr
(\mathtt{UR\_list} \ A \ B  \ (\lambda \ a \ b . a = \ \Transp{b}) \ l \ l' ) \Eqr 
(l={\Transp{l'}} )
$$

Note that it is always valid to decompose equality on inductive types. 
This is because a value of an inductive type can only be
observed by analyzing which constructor was used to build the value.
This fact is explicitly captured by the elimination principle of an
inductive type.
On the contrary, for dependent products, the fact that functions can
only be observed through application to a term is implicit in
\CIC, \ie~there is no corresponding elimination principle in the
theory (hence functional extensionality is an axiom).

The proofs that the constructors \coqe{nil} and \coqe{cons} are
univalently parametric are direct by definition of \coqe{UR_List}.
Likewise, the proof that the eliminator \coqe{list_rect} is univalently parametric is direct by induction on \coqe{UR_List}.

\paragraph*{Generalization}

It is possible to generalize the above result, developed for lists,
to any parameterized inductive family. As illustrated above, the univalent
relation for parameterized inductive families is given by
parametricity, and the proof that related inputs give rise to
equivalent types proceeds by a direct induction on the structure of
the type. The main difficulty is to generalize the proof of the
coherence of the relation with equality. Indeed, this involves fairly
technical reasoning on equality and injectivity of constructors.

Fortunately, in practice in our Coq framework, a general
construction is not required to handle each new inductive type,
because a witness of the fact that a given inductive is univalently parametric can be defined specifically as a typeclass instance.
We also provide a tactic to automatically generate this proof on any parameterized datatype (up to a fixed number of constructors), depending on the univalent parametricity of its parameters.

\subsection{Indexed Inductive Families}
\label{sec:gadts}
CIC supports the definition of inductive types that are not only parameterized, but also indexed, like length-indexed vectors $\mathtt{Vector} \ A \ n$. 
Another mainstream example is Generalized
Algebraic Data Types (GADTs) \cite{peytonJonesAl:icfp2006} illustrated here with the typical application to modeling typed expressions:
\begin{shaded}
\begin{coq}
Inductive Expr : Type -> Type :=
| I : nat -> Expr nat
| B : bool -> Expr bool
| Ad : Expr nat -> Expr nat -> Expr nat
| Eq : Expr nat -> Expr nat -> Expr bool.
\end{coq}
\end{shaded}
Observe that the return types of constructors instantiate the inductive
family at specific type indices, instead of uniform type parameters as is the case for
\eg~the parameterized list inductive type. This specificity of constructors is exactly what 
makes GADTs interesting for certain applications; but this is
precisely why their univalent parametricity is ineffective!

Indeed, consider the equivalence between natural numbers \coqe{nat} and binary natural numbers \coqe{N}. Univalent parametricity of the \coqe{Expr} GADT means that \coqe{Expr nat} is equivalent to \coqe{Expr N}. However, there is no constructor for \coqe{Expr} that can produce a value of type \coqe{Expr N}. So the only way to obtain such a term is by using an {\em equality} between \coqe{nat} and \coqe{N}, that is, using the univalence axiom.\footnote{It is however impossible to {\em prove} that no term of type \textsf{Expr Bin} can be constructed without univalence, because the univalence axiom is compatible with \CIC.}

The challenge is that univalent parametricity for indexed inductive families relies on
the coherence condition. To better understand this point, let us study
the prototypical case of identity types.

\paragraph{Identity types}

In Coq, the identity (or equality) type $\mathtt{eq}$, with
notation $=$, is defined as an indexed 
inductive family with a single constructor \coqe{eq_refl}:
\begin{shaded}
\begin{coq}
Inductive eq (A : Type) (x : A) : A -> Type :=  eq_refl : x = x.
\end{coq}
\end{shaded}
 The elimination principle \coqe{eq_rect}, known as {\em path induction} in HoTT terminology, is:
\begin{shaded}
\begin{coq}
eq_rect : forall A (x : A) (P : forall a : A, x = a -> Type), P x eq_refl -> forall (y : A) (e : x = y), P y e
\end{coq}
\end{shaded}

\begin{proposition} \label{prop:FP_Eq}
  There is a term
  $$\Univ{\mathtt{eq}} : \heteq{\mathtt{eq}}{\mathtt{eq}}{\Pi \ (A :
    \Type) \ (x : A). \ A -> \Type} {\Pi \ (A :
    \Type) \ (x : A). \ A -> \Type}.$$
\end{proposition}
\begin{proof}

  Unfolding the definition, this amounts to defining an inhabitant of
  $\uparamT{\Type}$ given two types $A$, $B$ related by
  $\UR{AB} : \heteq{A}{B}{\Type}{\Type}$, two terms $a$ in $A$, $b$ in
  $B$ related by $e: \heteq{a}{b}{A}{B}$ and two terms $a'$ in $A'$,
  $b'$ in $B$ related by $e': \heteq{a'}{b'}{A}{B}$.
  
The univalent relation for identity types is defined using the inductive type obtained by applying parametricity to the identity type: 
\begin{shaded}
\begin{coq}
Inductive UR_eq (A_1 A_2 : Type) (A_R : A_1 -> A_2 -> Type) (x_1 : A_1) (x_2 : A_2) (x_R : A_R x_1 x_2) :
   forall (y_1 : A_1) (y_2 : A_2), A_R y_1 y_2 -> x_1 = y_1 -> x_2 = y_2 -> Type :=
   UR_eq_refl : UR_eq A_1 A_2 A_R x_1 x_2 x_R x_1 x_2 x_R eq_refl eq_refl.
\end{coq}
\end{shaded}
The univalent relation is just a specialization of \texttt{UR\_eq}
where $\mathtt{A}_R$ is given by the relation induced by $\UR{AB}$ and
with $e$ and $e'$, so
we set:
\begin{align*} 
R_\mathtt{eq} &\defeq \lambda (e_1 : a =_A a') (e_2 : b =_B b') . \ 
  \mathtt{UR\_eq} \ A \ B \ \proji{\UR{AB}} \ a \ b \ e \ a' \ b' \ e' \ e_1 \ e_2 
\end{align*}

To prove that $\Eq{(a = a')}{(b = b')}$, it is necessary to use the
coherence between the relation and the equivalence provided by
$\UR{AB}$.
After rewriting, the equivalence to establish is
$$
\Eq{(a =_A a')}{(\Transp{a} =_B \Transp{a'})}
$$
This equivalence is again similar to a standard result of HoTT, namely
\coqe{equiv_functor_eq} in the Coq HoTT
library.

Finally, proving that the relation is coherent with equality amounts to show that 
$$
 \Pi (e_1 \ e_2 : a = a') . \ \Eq{(e = e')}
  {\mathtt{UR\_eq} \ A \ B \ \proji{\UR{AB}} \ a \ b  \ e  \ a' \ b' \
    e' \ e_1 \Transp{e_2}}
$$
This can be done by first showing the following equivalence\footnote{The notation $\transEq{e}{t}$, with $e:x=y$ and $t : P \ x$ when $P$ 
is clear from the context, denotes the transport of the term $e$ through the  equality proof $e$ (hence $\transEq{e}{t} : P \ y$).}
$$\mathtt{UR\_eq} \ A \  B \ P \ x \ y \ H \ x' \ y' \ H' \ X \ Y
\Eqr
 (\transEq{Y}{\;(\transEq{X}{H})} = H')$$
which means that the naturality square between $H$ and $H'$ commutes. 
\end{proof}

The proofs that \coqe{eq_refl} and \coqe{eq_rect} are univalently parametric are
direct by \texttt{UR\_eq\_refl} and elimination of $\mathtt{UR\_eq}$.

To deal with other indexed inductive types, one can follow a similar approach. Alternatively, it is possible to exploit the correspondence between an indexed inductive family and a subset of parameterized
inductive family, established by \citet{gambino2004wellfounded}, to prove the univalence of an indexed inductive family. In this correspondence, 
the property of the subset type is obtained from
the identity type.

For instance, for vectors:
$$
\Eq{\mathtt{Vector} \ A \ n}{\Sigma \ l : \mathtt{list} \ A. \ \mathtt{length} \ l = n}
$$
The \coqe{length} function computes the length of a
list, as follows:
\begin{shaded}
\begin{coq}
Definition length {A} (l: list A) : nat := list_rect A (fun _ => nat) O (fun _ l n => S n) l
\end{coq}
\end{shaded}
\noindent
where one can observe that the semantics of the index in the different constructors of vectors is captured in the use of the recursion principle \coqe{list_rect}.
By the abstraction theorem,
\mbox{$\Sigma \ l : \mathtt{list} \ A. \ \mathtt{length} \ l = n$} is
univalently parametric, and thus so is $\mathtt{Vect} \ A \ n$.

\section{Univalent Parametricity in Coq}
\label{sec:coq}

The whole development of univalent parametricity exposed in this
article has been formalized and implemented in the Coq 
proof assistant~\cite{Coq:manual}, reusing several constructions from the HoTT library
\cite{Bauer:2017:HLF:3018610.3018615}. We have formalized in \Coq the 
univalent parametricity translation, in order to mechanically verify the 
content from \S\ref{sec:univparam}---we do not discuss this effort here. 
Instead, we present the shallow embedding of the univalent relation in \Coq, based on typeclass instances to define and automatically derive the univalent parametricity proofs of Coq constructions. This framework brings the benefits of univalent parametricity to standard \Coq developments.

We first introduce the core classes of the
framework (\S\ref{sec:coq-framework}), and then describe the
instances for some type constructors
(\S\ref{sec:coq-constructors}). 
We explain how the use of typeclasses gives rise to
a direct implementation of univalent transport {\`a} la carte (\S\ref{sec:tr-carte}).
Finally, we discuss several refinements to the framework to circumvent the limitation of relying on the univalence axiom in Coq (\S\ref{sec:effectiveness}).

\subsection{Coq Framework}
\label{sec:coq-framework}

The central notion at the heart of this work is that of type equivalences, which we formulate as a typeclass to allow automatic inference of equivalences:\footnote{Adapted
  from:
  \url{http://hott.github.io/HoTT/coqdoc-html/HoTT.Overture.html}.}
\begin{shaded}
\begin{coq}
Class IsEquiv (A B : Type) (f : A -> B) := {
  e_inv : B -> A;
  e_sect : forall x, e_inv (f x) = x;
  e_retr : forall y, f (e_inv y) = y;
  e_adj   : forall x, e_retr (f x) = ap f (e_sect x) }.
\end{coq}
\end{shaded}
 The properties \coqe{e_sect} and \coqe{e_retr} express that
\coqe{e_inv} is both the left and right inverse of \coqe{f}, respectively.
The property \coqe{e_adj} is a compatibility
condition between the proofs. It ensures that
the equivalence is uniquely determined by the function \coqe{f}. 

While \coqe{IsEquiv} characterizes a particular function $f$ as being an equivalence, we say that two types $A$ and $B$ are equivalent, noted $\Eq{A}{B}$, iff there exists such a function $f$.
\begin{shaded}
\begin{coq}
Class Equiv A B := { e_fun :> A -> B ; e_isequiv : IsEquiv e_fun }.
Infix "≃" := Equiv.
\end{coq} 
\end{shaded}
\coqe{Equiv} is here defined as a typeclass to allow automatic inference of equivalences. This way, we can define automatic univalent transport as follows:
\begin{shaded}
\begin{coq}
Definition univalent_transport {A B : Type} {e : A ≃ B} : A -> B := e_fun e.  
Notation "upaBB" := univalent_transport.
\end{coq}
\end{shaded}
\noindent where the equivalence is obtained through typeclass instance resolution, \ie~proof search. Note also that the source and target types of transport are implicitly resolved by default.

To formalize univalent relations, we define a hierarchy of classes, starting from \coqe{UR} for univalent relations (arbitrary heterogeneous relations), refined by \coqe{UR_Coh}, which additionally requires the proof of coherence between a univalent relation and equality.

\begin{shaded}
\begin{coq}
Class UR A B := { ur : A -> B -> Type  }.
Infix "≈" := ur.

Class UR_Coh A B (e : Equiv A B) (H : UR A B) := { ur_coh : forall (a a' : A), Equiv (a = a') (a ≈ upaBB a') }.
\end{coq}
\end{shaded}
 The attentive reader will notice that the definition of the coherence condition above is dual to the one stated in Figure~\ref{fig:univ}. Both definitions are in fact equivalent.\footnote{See lemma \coqe{is_equiv_alt_ur_coh} in the Coq development.} The reason for adopting this dual definition is that it eases the definition of new instances.

As presented in Figure~\ref{fig:univ}, two types are related by the
univalent parametricity relation if they are equivalent and there is a
coherent univalent relation between them. 
This is captured by the typeclass \coqe{UR_Type}. 

\begin{shaded}
\begin{coq}
Class UR_Type A B := { 
    Ur :> UR A B;
    equiv :> A ≃ B;
    Ur_Coh :> UR_Coh A B equiv Ur;
    Ur_Can_A :> Canonical_eq A;
    Ur_Can_B :> Canonical_eq B }.
Infix "papillon" := UR_Type.
\end{coq}
\end{shaded}
 The last two attributes are part of the Coq framework in order to better support extensibility and effectiveness, as will be described in \S\ref{sec:canon-equal-types}.

\subsection{Univalent Type Constructors}
\label{sec:coq-constructors}

The core of the development is devoted to the proofs that standard type
constructors are univalently parametric, notably $\Type$ and $\Pi$. In terms of the Coq framework, this means providing \coqe{UR_Type} instances relating each constructor to itself. These instance definitions follow directly the proofs discussed in \S\ref{sec:univparam}.

For the universe $\Type_i$, we define:

\begin{shaded}
\begin{coq}
Instance UR_Type_def@{i j} : UR@{j j j} Type@{i} Type@{i} := {| ur := UR_Type@{i i i} |}.
\end{coq}
\end{shaded}

This is where our fixpoint construction appears: the relation at
$\Type_i$ is defined to be \coqe{UR_Type} itself. So, for a type to be
in the relation means more than mere equivalence: we also get a relation
between elements of that type that is coherent with equality. This
\coqe{UR_Type_def} instance will be used implicitly everywhere we use
the notation \coqe{X ≈ Y}, when \coqe{X} and \coqe{Y} are types
themselves.

For dependent function types, we set:
\begin{shaded}
\begin{coq}
Definition UR_Forall A A' (B : A -> Type) (B' : A' -> Type) (dom : UR A A')
  (codom: forall x y (H : x ≈ y), UR (B x) (B' y)) : UR (forall x, B x) (forall y, B' y) :=
    {| ur := fun f g => forall x y (H: x ≈ y), f x ≈ g y |}.
\end{coq}
\end{shaded}

The univalent parametricity relation on dependent function types expects
relations on the domain and codomain types, the latter being
parameterized by the former through its argument \coqe{(H : x ≈ y)}.
The definition is the standard extensionality principle on dependent
function types.

Then, we need to show that the universe is related to itself
according to the relation on the universe at one level above.
This corresponds to Proposition~\ref{prop:FP_Type}, and thanks to
both universe polymorphism and implicit management of universe
levels in \Coq, there is no need to mention levels at all (note that
\coqe{Type ≈ Type} in the definition actually computes to \coqe{Type ⋈ Type}):\footnote{ 
Some universe annotations appear in the Coq source files in order to explicitly validate our assumptions about universes.}

\begin{shaded}
\begin{coq}
Definition FP_Type : Type ≈ Type.
\end{coq}
\end{shaded}

Regarding the dependent product, the \coqe{Equiv} instance that needs
to be defined in order to show that it is univalently parametric has the following type:

\begin{shaded}
\begin{coq}
Instance Equiv_forall : forall (A A' : Type) (e_A : A ≈ A') (B : A -> Type) 
    (B' : A' -> Type) (e_B : B ≈ B'),  (forall x : A, B x) ≃ (forall x : A', B' x).
\end{coq}
\end{shaded}

While the conclusion is an equivalence, the assumptions \coqe{e_A} and
\coqe{e_B} are about univalent relations for \coqe{A}, \coqe{A'} and
\coqe{B} and \coqe{B'}. The first one is implicitly resolved as the
\coqe{UR_Type_def} defined above, and the second one as a combination
of \coqe{UR_Forall} and \coqe{UR_Type_def}. With these stronger
assumptions, and because \coqe{≈} is heterogeneous, we can prove the
equivalence without introducing transports.
This is key to make the typeclass instance proof search tractable: it
is basically structurally recursive on the type indices. We can then
show that the dependent function type seen as a binary type
constructor is related to itself using the univalent relation and
equivalence constructed above, and the coherence proof
$\univTerm{\Pi}$ presented in the Abstraction
Theorem (Th.~\ref{thm:abstraction}):

\begin{shaded}
\begin{coq}
Definition FP_forall : (fun A B => (forall x : A , B x)) ≈ (fun A' B' => (forall x : A', B' x)).
\end{coq}
\end{shaded}
To instrument the typeclass instance proof search, we add proof
search \coqe{Hint}s for each fundamental property.

We proceed similarly for other constructors: dependent pairs, the identity type, natural numbers and booleans with the canonical univalent relation, where we additionally prove the fundamental property for the eliminators. That is to say, we have many fundamental property lemmas such as:  
\begin{shaded}
\begin{coq}
Definition FP_Sigma : @sigT ≈ @sigT.
\end{coq}
\end{shaded}

Having spelled out the basics of the Coq framework for univalent parametricity, we can now turn to the practical issue of effective transport. 

\subsection{Univalent Transport {\`a} la Carte}
\label{sec:tr-carte}

Because the computation of the univalent parametricity relation is done using type
class resolution, the framework is already set up to support extension
with contexts of univalently-related constants, using typeclass instances.
For example, to extend the context with the fact that two types \coqe{A}
and \coqe{B} are univalently related, one needs to
provide a proof that there is an equivalence between \coqe{A} and \coqe{B}, and
declare it as a typeclass instance. Then, using for example the
canonical relation \mbox{\coqe{fun (a : A) (b : B) => a = upaBB b}}, one can
construct an instance of \coqe{A ⋈ B}.
The user also needs to define an instance of \coqe{B ⋈ A} using the symmetry
of the relation, because typeclass resolution cannot be instrumented
with a general rule for symmetry, otherwise proof search would never
terminate.
This way, typeclass resolution is able to automatically derive
further instances of the relation based on this univalent relation instance.

Note that the relation between \coqe{A} and \coqe{B} does not have to
be the canonical relation. For instance, coming back to the example of
the equivalence between sized lists and vectors, the relation between
sized lists and vectors can be based either on equality (plus transport
using the equivalence), on the relation \coqe{UR_list} (plus transport
using the equivalence), on a similar relation on vectors (plus transport
using the equivalence) or even on a heterogeneous relation directly
relating sized lists and vectors (without the use of transport).
Of course, all these definitions are equivalent because of the
coherence condition of univalent parametricity, but they can have
different computational content and a user can favor one or the
other, depending on the application in mind. 

Going one step further, one can show that two functions defined on
univalently-related types are in univalent relation. Let us say for
example that \coqe{f : A -> A -> A} and \coqe{g : B -> B -> B} are in
univalent relation, with witness \coqe{univrel_fg}.
Then, to exploit this univalent relation to perform univalent transport
\`a la carte, the user needs to add the following \coqe{Hint} to the proof
search database:\footnote{We use hint declarations to have precise
  control over the shape of goals a typeclass instance should solve and how.}
\begin{shaded}
\begin{coq}
Hint Extern 0 (f _ _ ≈ g _ _) => eapply univrel_fg : typeclass_instances.
\end{coq}
\end{shaded}
\noindent This declaration amounts to extending the global context of univalently-related constants with \coqe{f} and \coqe{g}.

In general, when adding new instances to univalent parametricity, the
user needs to define corresponding \coqe{Hint}s to enable automatic
typeclass resolution.
\S\ref{sec:extended-example} explains in more details how the
examples of \S\ref{sec:univ-param-acti} are realized, including such hints.

\subsection{Effectiveness of Univalent Parametric Transport in Coq}
\label{sec:effectiveness}

The proofs of univalent parametricity we have developed in \S\ref{sec:univparam} are in a setting where univalence is realized as an axiom (\CC and CIC). The axiomatic nature of the development manifests as follows:
\begin{enumerate}
\item The univalence axiom proper is used to show the coherence condition of univalent parametricity for the universe (\S\ref{sec:univ_type}). This is to be expected and unavoidable, as this condition for the universe exactly states that type equivalence coincides with equality. 
\item Functional extensionality (an axiom in CIC, which follows from univalence) is used to show that the transport functions of the equivalence for the dependent product form an equivalence (\S\ref{sec:translation}).
\end{enumerate}

Additionally, as shown in \S\ref{sec:inductives}, the effectiveness of univalent transport for inductive types depends on the type of parameters and indices. In particular, proving univalent parametricity of indexed families requires using the coherence condition.

To see how this relates to practice, consider the case of functions. Functional extensionality is only used in the proof that the transport functions form an equivalence. In particular, this means that the transport functions themselves {\em are} effective. Therefore, when transporting a first-order function, the resulting function is effective. 

For the axiom to interfere with effectiveness, we need to consider a {\em higher-order} function, \ie~that takes another function as argument. 
Consider for instance the conversion of a higher-order dependent function \coqe{g} operating on a function over natural numbers
\begin{shaded}
\begin{coq}
  g : forall (f : nat -> nat), Vector nat (f O)
\end{coq}
\end{shaded}
 to one operating on a function over binary natural numbers
\begin{shaded}
\begin{coq}
 g' : forall (f : N -> N), Vector N (f upaBB O) := upaBB g.
\end{coq}
\end{shaded}
We transport \coqe{g} to \coqe{g'} along the equivalence between the two higher-order types above. Such a transport uses, {\em in a computationally-relevant position}, the fact that the function argument \coqe{f} can be transported along the equivalence between \coqe{nat -> nat} and \coqe{N -> N}. Consequently, the use of functional extensionality in the equivalence proof chimes in, and \coqe{g'} is not effective. (Specifically, \coqe{g'} pattern matches on an equality between natural numbers that contains the functional extensionality axiom.)

Fortunately, there are different ways to circumvent this problem, by
exploiting the fact that univalent parametricity is specializable through specific typeclass
instances. This section shows how we can further specialize proofs of
univalent parametricity in situations where using axioms can be
avoided. Sometimes we can ignore the fact that an equality proof might
be axiomatic by automatically crafting a new one that is axiom-free
(\S\ref{sec:canon-equal-types}), or we can avoid transporting
type families with (potentially axiomatic) proofs of equality in some
 cases (\S\ref{sec:transportable}).

\subsubsection{Canonical Equality for Types with Decidable Equality}
\label{sec:canon-equal-types}

Any proof of equality between two natural numbers can be turned into a canonical, axiom-free proof using decidability of equality 
on natural numbers.
In general, decidable equality on a type \coqe{A} can be expressed in type
theory as
\begin{shaded}
\begin{coq}
Definition DecEq (A : Type) := forall x y : A, (x = y) + nnot(x = y).
\end{coq}
\end{shaded}%
Hedberg's theorem~\cite{hedberg:jfp1998} implies that if \coqe{A} has decidable equality, then \coqe{A} satisfies Uniqueness of Identity Proofs (UIP): 
any two proofs of the same equality between elements of \coqe{A} are
equal.
Hedberg's theorem relies on the construction of a canonical equality
to which every other is shown equal. Specifically, when \coqe{A} has a decidable equality, it is possible to define a function
\begin{shaded}
\begin{coq}
Definition Canonical_eq_decidable A (Hdec : DecEq A) : forall x y : A, x = y -> x = y :=
  fun x y e => match Hdec x y with
               | inl e0 => e0
               | inr n => match (n e) with end end.
\end{coq}
\end{shaded}%
This function produces an equality between two terms \coqe{x} and
\coqe{y} of type \coqe{A} by using the decision procedure \coqe{Hdec},
independently of the equality \coqe{e}.
In the first branch, when \coqe{x} and \coqe{y} are equal, it returns the 
canonical proof produced by \coqe{Hdec}, instead of propagating the input (possibly-axiomatic) proof \coqe{e}. And in case the decision procedure returns an inequality proof (of type \coqe{x=y -> False}), the function uses \coqe{e} to establish the contradiction.
In summary, the function transforms any equality into a canonical equality by using the input equality \emph{only in cases that are not possible}.

We can take advantage of this insight to ensure effective transport on indices of types with decidable equality. The general idea is to extend the relation on types $\heteqr{A}{B}$ to also include two functions
\coqe{forall x y : A , x = y -> x = y} and 
\coqe{forall x y : B , x = y -> x = y}. For types with decidable equality, these functions can exploit the technique presented above, and for others, these are just the identity.
However, care must be taken: we cannot add arbitrary new computational content
to the relation; we have to require that these functions preserve reflexivity. This is specified in the following class: 
\begin{shaded}
 \begin{coq}
 Class Canonical_eq (A : Type) :=
   { can_eq : forall (x y : A), x = y -> x = y ;
     can_eq_refl : forall x, can_eq x x eq_refl = eq_refl }.
\end{coq}
\end{shaded} %
\noindent which is used for the last two attributes of the \coqe{UR_Type} class given in \S\ref{sec:coq-framework}.

There are two canonical instances of \coqe{Canonical_eq}, the one that 
is defined on types with decidable equality, and exploits the technique above, and the default one, which is given by the identity function (and proof by reflexivity).

Using this extra information, it is possible to improve the definition
of univalent parametricity by always working with canonical
equalities. This way, equivalences for inductive types whose indices are of types with decidable equality---like length-indexed vectors and many common examples---never get stuck on rewriting of indices.

\subsubsection{Canonically-Transportable Predicates}
\label{sec:transportable}

As mentioned in the introduction of this section, for some predicates, 
it is not necessary to pattern match on equality to implement transport.

The simplest example is when the predicate does not actually depend on
the value, in which case \coqe{P x ≃ P y} can be implemented by the
identity equivalence because \coqe{P x} is convertible to \coqe{P y},
independently of what \coqe{x} and \coqe{y} are. 
It is also the case when the predicate is defined on a type with a
decidable equality, so we can instead pattern match on the canonical
equality (\S\ref{sec:canon-equal-types}). 

To take advantage of this situation whenever possible, we introduce
the notion of \emph{transportable} predicates.
\begin{shaded}
\begin{coq}
Class Transportable {A} (P : A -> Type) := {
    transportable :> forall x y, x = y -> P x ≃ P y;
    transportable_refl : forall x, transportable x x eq_refl = Equiv_id (P x) }.
\end{coq}
\end{shaded}%
Note that as for \coqe{Canonical_eq}, we need to require that
\coqe{transportable} behaves like the standard transport of equality
by sending reflexivity to the identity equivalence.  

For instance, the instance for constant type-valued functions is
defined as 
\begin{shaded}
\begin{coq}
Instance Transportable_cst A B : Transportable (fun _ : A => B) := {|
    transportable := fun (x y : A) _ => Equiv_id B;
    transportable_refl := fun x : A => eq_refl |}.
\end{coq}
\end{shaded}
To propagate the information that every predicate (a.k.a. type family)
comes with its instance of \coqe{Transportable}, we specialize the
definition of \coqe{UR (A -> Type) (A' -> Type)}:
\begin{shaded}
\begin{coq}
Class URForall_Type_class A A' {dom : UR A A'}  (P : A -> Type) (Q : A' -> Type) :=
  { transport_ :> Transportable P; ur_type :> forall x y (H:x ≈ y), P x ⋈ Q y }.

Definition URForall_Type A A' {HA : UR A A'} : UR (A -> Type) (A' -> Type) :=
    {| ur := fun P Q => URForall_Type_class A A' P Q |}.
\end{coq}
\end{shaded}%
This definition says that two predicates are in relation whenever they are
in relation pointwise, and when \coqe{P} is transportable.

Using \coqe{Transportable}, we can instrument the definition of
univalent relation on dependent products to improve effectiveness.
More precisely, in the definition of the inverse function that defines
the equivalence
\coqe{(forall x : A, B x) ≃ (forall x : A', B' x)}
we use the fact that \coqe{B} is transportable to change the dependency
in \coqe{B} instead of pattern matching on the equality between the
dependencies.
This is possible because from \coqe{e_B : B ≈ B'}, we know that
\coqe{B} is transportable (thanks to the specialized definition
\coqe{URForall_Type}).

\section{Univalent Parametricity in Action: Explained}
\label{sec:extended-example}

We now come back to the examples of \S\ref{sec:univ-param-acti},
explaining the definitions and adjustments of the typeclass resolution
mechanism necessary to achieve seamless transport {\`a} la carte.

\paragraph{Adding a univalent relation in the global context}
To declare that unary and binary natural numbers are univalently
related, one first needs to provide a proof \coqe{IsEquiv_of_nat} that
the transport function \coqe{N.of_nat} from \coqe{nat} to \coqe{N} is
actually an equivalence, and declare it as a typeclass instance.
\begin{shaded}
\begin{coq}
    Instance Equiv_nat_N : nat ≃ N := BuildEquiv nat N N.of_nat IsEquiv_of_nat.
\end{coq}
\end{shaded}
Then, using for example the canonical relation 
\coqe{fun (n:N) (m:nat) => n = upaBB m} to define the univalent relation 
\coqe{UR_nat_N}
between \coqe{nat} and \coqe{N}, the only remaining piece is the proof
of the following coherence condition, which can easily be done using
the section of the equivalence.
\begin{shaded}
\begin{coq}
Definition coherence_nat_N : forall a a' : nat, (a = a') ≃ (a = N.to_nat (N.of_nat a')).
\end{coq}
\end{shaded}
Then, one can define an instance of \coqe{nat ⋈ N}.
\begin{shaded}
\begin{coq}
  Instance univrel_nat_N : nat ⋈ N :=
  {| equiv := Equiv_nat_N;
     Ur := UR_nat_N;
     Ur_Coh := {| ur_coh := coherence_nat_N |}; |}
\end{coq}
\end{shaded}
This way, typeclass resolution is able to automatically derive
further instances of the relation based on this basic univalent
relation, emulating the White Box FP of Corollary~\ref{white box}.
For instance, \Coq can automatically infer the relation between \coqe{nat -> nat -> nat} and
\coqe{N -> N -> N} because it is equal to \coqe{upaWB (nat -> nat -> nat)}: 
\begin{shaded}
\begin{coq}
    Goal nat -> nat -> nat ⋈ N -> N -> N.
      typeclasses eauto.
    Qed.
\end{coq}
\end{shaded}

\paragraph{Transport {\`a} la carte}
However, the fact that \coqe{nat} and \coqe{N} are in univalent
relation alone only provides black-box transport on functions
manipulating integers. To get more efficient transport, one can do an
additional proof effort in order to also get some white-box transport,
exploiting the relation between some particular functions. For
instance, one can prove that the addition functions on unary and
binary numbers are univalently related:
\begin{shaded}
\begin{coq}
Definition  univrel_add : plus ≈ plus_N.
\end{coq}
\end{shaded}
This amounts to show that for every \coqe{n m : nat}, we have
\coqe{plus n m = upaBB (plus_N (upaBB n) (upaBB m))}.
The proof can be done by induction on \coqe{n}. The \coqe{O} case
requires showing that \coqe{plus_N N0 m = m} for every \coqe{m : N},
which is true by computation. The \coqe{S} case requires showing that
\coqe{plus_N (NS n) m =} \mbox{\coqe{NS (plus_N n m)}} for every \coqe{n m: Bin}. This
property is more complex to prove because it must be done by induction
on \coqe{n} and the definition of \coqe{NS} does not comply very well
with the binary structure.

Next, to add this relation to the global context, we need to
instrument typeclass resolution by defining the following hint, which
will be used when looking for a function in relation with the
\coqe{plus} function:
\begin{shaded}
\begin{coq}
Hint Extern 0 (plus _ _ ≈ _) => eapply univrel_add : typeclass_instances.
\end{coq}
\end{shaded}
With this hint, the system is able to automatically infer that the
types \coqe{forall n m : nat, n + m = m + n} and
\coqe{forall n m : N, n + m = m + n} are in univalent relation using
white-box transport. For instance:
\begin{shaded}
\begin{coq}
Goal (forall n m : nat, n + m = m + n) ⋈ (forall n m : N, n + m = m + n).
  typeclasses eauto.
Qed.
\end{coq}
\end{shaded}
It is therefore possible to automatically transport the {\em proof term} of the
commutativity of \coqe{plus} to a proof term of the commutativity of
\coqe{plus_N} using the black-box transport provided by this univalent relation.
\begin{shaded}
\begin{coq}
Definition plus_N_comm : forall n m : N, n + m = m + n := upaBB plus_comm.
\end{coq}
\end{shaded}

\paragraph{Transporting goals}
As explained in \S\ref{sec:univ-param-acti}, univalent
parametricity can also be used to prove
properties by computation using an alternative representation that is more adequate computationally. For instance with the polynomial
\coqe{poly}, the proof that \coqe{poly 50} is bigger than \coqe{1000}
can be done by moving to an equivalent property on binary natural
numbers first, and then solving the goal by computation and basic inversion. 
\begin{shaded}
\begin{coq}
Goal poly 50 >= 1000.
  replace_goal; now compute.
Defined.
\end{coq}
\end{shaded}
The tactic \coqe{replace_goal} proceeds by first asserting
that there exists a property \coqe{opt} that is in univalent relation with the
given goal (here \coqe{poly 50 >= 1000}), and inferring this
equivalent property through typeclass resolution using white-box
transport.
Then the equivalence induced by the univalent relation is used to
replace the original goal with the inferred property \coqe{opt}; this
is black-box transport.
The definition of the \coqe{replace_goal} tactic is simple in
\coqe{Ltac}:
\begin{shaded}
\begin{coq}
Ltac replace_goal :=
  match goal with | |- ?P => let X := fresh "X" in
    refine (let X := _ : { opt : Prop & P ≈ opt} in _);
      [ eexists; typeclasses eauto | apply (e_inv (equiv X.2))]
  end.
\end{coq}
\end{shaded}
It first introduces the definition of a property \coqe{opt} that is in
univalent relation with \coqe{P}. This property \coqe{opt} is obtained automatically by triggering the typeclass resolution on \coqe{P ≈ ?}, finding a canonical instance. If it succeeds, this step gives at the same time the definition of \coqe{opt} and the proof that it is in univalent relation with \coqe{P}. In particular, this proof contains an equivalence between \coqe{P} and
\coqe{opt}, which is used to replace the goal \coqe{P} by \coqe{opt}, using
the inverse function of the equivalence. 

\paragraph{Fixpoints.} 
The proof technique above also scales to fixpoints, even though fixpoints must be dealt with in a non-generic way. Concretely, one needs to provide a univalent
parametricity instance for each case of pattern matching performed
inside the fixpoint. In the sequence example of \S\ref{sec:auto-comp}, the required instance must be defined for
fixpoint matching on \coqe{0}, \coqe{1}, and \coqe{2}:  
\begin{shaded}
\begin{coq}
  Definition fix_nat_3 :
    (fun P X0 X1 X2 XS => fix f (n : nat) {struct n} : P := match n with
                                                                                           | 0 => X0
                                                                                           | 1 => X1
                                                                                           | 2 => X2
                                                                                           | S n => XS n (f n) end) ≈
    (fun P X0 X1 X2 XS => fix f (n : nat) {struct n} : P := match n with
                                                                                           | 0 => X0
                                                                                           | 1 => X1
                                                                                           | 2 => X2
                                                                                           | S n => XS n (f n) end).
\end{coq}
\end{shaded}
The proof of this instance is systematic, and can be done
automatically using \coqe{induction} and \coqe{typeclasses eauto}.
However, because pattern matchings are not first-class
objects in \Coq, it is not possible to define a single generic univalent parametric instance for every fixpoint.
Note that this (practical, rather than theoretical) issue does not
manifest when using eliminators, because there is only one eliminator
per inductive type.

\paragraph*{Limitations of the current \Coq implementation} 

The use of typeclasses and typeclass resolution to deal with the global
context of univalently related constants is both a blessing
and a curse.
It is nice because it allows us to instrument univalent parametricity in
\Coq without modifying the source code, offering great flexibility and accessibility.  
But this approach does not scale very well to large developments because
typeclass resolution is internally based on proof search, which
quickly becomes intractable. 
In practice, in our current implementation, we observe that successful typeclasss resolution is fairly fast, but when the proof
search fails because of some missing \coqe{Hint}s, resolution can take a
very long time or may even diverge. 

This issue is a known limitation of using typeclasses to drive automatic program transformations, and can also be experienced in other frameworks like CoqEAL~\cite{cohenAl:cpp2013}. 
It could be addressed via a direct implementation of univalent parametricity, for instance using MetaCoq~\cite{sozeau:hal-02167423,MetaCoq:popl2020}. With a MetaCoq plugin, it is possible to have
complete access to the reification of a term of \Coq in \Coq. 
This would provide complete programmatic control over the univalent
parametricity  translation, thereby avoiding issues that follow from
relying on proof search.
Another possibility is to implement the translation directly as a
\Coq plugin in \OCaml, as has recently been done for the white-box approach by \citet{ringer_et_al:LIPIcs:2019:11081}. However, the direct
definition of a plugin is very sensitive to changes in the
implementation of the \Coq proof assistant itself, so we believe that the
MetaCoq approach would be better suited, as it provides an abstraction barrier between the theory of \Coq and its actual implementation.

\section{Case Study: Native Integers}
\label{sec:ffi}

To further illustrate the applicability of univalent parametricity, we consider a case study based on a recent improvement to \Coq: native 63-bits integers, available starting with \Coq~8.10.\footnote{See \url{https://github.com/coq/coq/blob/v8.10/theories/Numbers/Cyclic/Int63/Int63.v}}
This extension raises the question of how to interface a native datatype within \Coq, supporting reasoning about (and with) such native values.

Native integers provide a basic datatype \coqe{int} together with
native functions such as the left lshift operator \coqe{a << b}, which
shifts each bit in \coqe{a} to the left by the number of
positions indicated by \coqe{b}. These are defined as follows:
\begin{shaded}
\begin{coq}
Register int : Set as int63_type.
Primitive lsl := #int63_lsl.
Infix "<<" := lsl (at level 30, no associativity) : int63_scope.
\end{coq}
\end{shaded}
Because the operations are native, there is no direct way to reason about them
in \Coq. This is why the standard library of \Coq relates \coqe{int}
to binary numbers \coqe{Z} (these are similar to \coqe{N}, but include
negative numbers), and states {\em axioms} to specify the behavior of
native functions.
\begin{shaded}
\begin{coq}
  Definition wB := 2 ^ 63.

  Definition to_Z : int -> Z := ... (* explicit definition using operations on int *)
  Definition of_Z : Z -> int := ... (* explicit definition using operations on int *)

  Axiom of_to_Z : forall (x:int), of_Z (to_Z x) = x.
  Axiom lsl_spec : forall x p, to_Z (x << p) = to_Z x * 2 ^ (to_Z p) mod wB.
\end{coq}
\end{shaded}
The statements of \coqe{of_to_Z} and \coqe{lsl_spec} are very natural,
but the first question it raises is about completeness: How can we be
sure that these two axioms are enough to prove any property on
\coqe{lsl}? For instance, do we also need to postulate that \coqe{to_Z} forms a
retraction?

Actually, it is possible to derive the other part of the
correspondence between \coqe{int} and \coqe{Z} (note that this is not
an axiom, it is proven by induction on \coqe{z} in \coqe{Z}):
\begin{shaded}
  \begin{coq}
    Lemma of_Z_spec : forall (z:Z), to_Z (of_Z z) = n mod wB.
\end{coq}
\end{shaded}

Considering \coqe{of_to_Z} and \coqe{of_Z_spec}, it would seem that \coqe{int} and \coqe{Z} are indeed univalently related and that functions on \coqe{int} can likewise be univalently related to functions on \coqe{Z}.
Actually, the careful reader might have noticed that \coqe{of_Z_spec}
does not exactly correspond to the statement of a retraction on
\coqe{to_Z}.
This is because \coqe{int} is actually in relation with
$\mathbb{Z}/2^{63}\mathbb{Z}$.
Therefore, we can define $\mathbb{Z}/2^{63}\mathbb{Z}$ as the type \coqe{ZwB}
of integers between $0$ and $2^{63}$, and adjust \coqe{to_Z} and
\coqe{of_Z} accordingly.
\begin{shaded}
\begin{coq}
  Definition ZwB := { n : Z & 0 <= n < wB }.

  Lemma to_Z_bounded : forall x, 0 <= to_Z x < wB.
  Definition to_ZwB : int63 -> ZwB := fun x => (to_Z x; to_Z_bounded x).
  Definition of_ZwB (z:ZwB) : int63 := of_Z z.1.
\end{coq}
\end{shaded}
The axioms \coqe{of_to_Z } and \coqe{lsl_spec} are exactly what is
required to show that \coqe{int} and \coqe{ZwB} are univalently
related and that the native function \coqe{lsl} is univalent
related to the corresponding function on \coqe{ZwB}.
\begin{shaded}
\begin{coq}
  Definition IsEquiv_to_Z_ : IsEquiv to_ZwB.
  Instance univrel_int_ZwB : int ⋈ ZwB.

  Definition ZwB_lsl : ZwB -> ZwB -> ZwB :=
    fun n m => ((n.1 * Z.pow 2 m.1) mod wB ; (* easy proof term omitted *)).

  Notation "n << m" := (ZwB_lsl n m) : ZwB_scope.

  Definition univrel_lsl : lsl ≈ ZwB_lsl.
\end{coq}
\end{shaded}
We have illustrated the correspondence between \coqe{int} and
\coqe{ZwB} using the \coqe{lsl} function, but the very same can be
done for all functions of the native integers interface.  

Once this is done, it is possible to use univalent parametricity to go
beyond what is currently provided in the \Coq standard library, such
as proving concrete properties on \coqe{ZwB} using an automatic
transport to \coqe{int}. Consider the following polynomial, and two similar proofs by computation: the first directly on \coqe{ZwB}, and the second on \coqe{int} after transport.
\begin{shaded}
\begin{coq}
  Definition poly_Z : ZwB -> ZwB :=  fun n => 45 + ZwB_pow n 100 - ZwB_pow n 99 * 16550.

  Goal poly_Z 16550 = 45.
    Time reflexivity. 
  Defined.

  Goal poly_Z 16550 = 45.
    replace_goal. Time reflexivity. 
  Defined.
\end{coq}
\end{shaded}
While both executions of \coqe{reflexivity} terminate, the execution
time when the goal is not shifted to \coqe{int} is two orders of magnitude slower than when it is ($0.3s$ vs. $0.002 s$).  
The difference for this precise (artificial) example may not seem
that significant in absolute terms, but we can expect it to be interesting in large-scale developments, which could justify the use of native integers.

The second---maybe more important---advantage of organizing all the
axioms on the specification of the functions on native integers using
univalent parametricity is that it guarantees {\em completeness} of
the axiomatization.
Indeed, by the White Box FP (Corollary~\ref{white box}), we are
certain that any theorem on \coqe{int} and its native functions is
univalently related to a theorem on \coqe{ZwB}.
And by the Black Box FP (Proposition~\ref{black box}), such
univalent relation allows us to easily transport the proof of this theorem on
\coqe{ZwB} to a proof of the corresponding theorem on \coqe{int}.
For instance, the distributivity of \coqe{<<} over addition can be automatically transported from \coqe{ZwB} to \coqe{int}:
\begin{shaded}
  \begin{coq}
  Definition ZwB_lsl_add_distr x y n : (x + y) << n = (x << n) + (y << n).
  (* proof using properties of mod and automation on Z *)

  Definition lsl_add_distr : forall x y n, (x + y) << n = (x << n) + (y << n) :=
    upaBB ZwB_lsl_add_distr.
\end{coq}
\end{shaded}
In contrast, the proof of \coqe{lsl_add_distr} in the \Coq
standard library is done manually, and in fact does not even use the auxiliary lemma
\coqe{ZwB_lsl_add_distr}. Instead, the proof is dealing with both the conversion to \coqe{Z} and the proof of the property on \coqe{Z} at the same time.
We believe that systematically proving properties first on \coqe{ZwB}, and
then automatically transporting them to \coqe{int} simplifies development, maintenance, and understanding.

\section{Related Work}
\label{sec:related}

\paragraph{Type theories.}
Homotopy Type Theory \cite{hottbook} treats equality of types as
equivalence. For regular datatypes (also known as homotopy sets or
hSets), equivalence boils down to isomorphism, hence the existence of
transports between the types. However, as univalence is considered as an
axiom, any meaningful use of the equality type to transport terms along
equivalences results in the use of a non-computational construction. In
contrast, in this work we carefully delimit the effective equivalence-preserving type
constructors in our setting, pushing axioms as far as possible, and supporting specialized proofs to avoid them in certain scenarios.

Cubical Type Theory~\cite{cubicaltt} provides computational content to
the univalence axiom, and hence functional and propositional
extensionality as well. In this case, the invariance of constructions
by type equivalence is built in the system and the equality type
reflects it.  Note that the recent work of
\citet{altenkirch15:towards} on a cubical type theory without an
interval proposes to use a relation quite close to the one defined in
univalent parametricity to encode equality in the theory. They are
handling a different problem (albeit in a similar way),
because they are trying to build a theory that supports univalence.
In our framework, we relate the relation to
equality and type equivalence, which allows us to stay within \CIC,
without relying on another more complex type theory, but the price we pay is to assume univalence as an axiom.
Note that while we focus on \CIC extended with the univalence axiom, univalent parametricity could be likewise developed in a cubical theory, such as Cubical Agda~\cite{vezzosiAl:icfp2019}. The only change is that the definition of univalent relations for types and functions can make use of top-level equality directly instead of using explicit extensional definitions such as type equivalence and pointwise equality. In particular, terms such as $\equivPi$ or $\univTerm{\Pi}$ (\S\ref{sec:univparam}) can be largely simplified. But the interest of the general relational univalent parametric setting remains unchanged because it provides heterogeneous automatic transport, which is not readily available in cubical type theories.

Observational Type Theory (OTT)~\cite{altenkirchAl:plpv2007} uses a
different notion of equality, coined John Major equality. It is a
heterogeneous relation, where terms in potentially
different types can be compared, 
usually with the assumption that the two types will
eventually be \emph{structurally} equal, not merely equivalent. This
stronger notion of equality of types is baked in the type system, where
type equality is defined by recursion on the structure of types, and value
equality follows from it. It implies the K axiom, which is in general
inconsistent with univalence, although certainly provable for all the
non-polymorphic types definable in OTT. A system similar to ours could
be defined on top of OTT to allow transporting by equivalences.

Parametric Type Theory and the line of work integrating parametricity
theory to dependent type theory, either internally
\cite{bernardyAl:entcs2015} or externally, is linked to the current work
in the sense that our univalent parametricity translation is a
refinement of the usual parametricity translation. 
In its simple form, parametricity in type theory does not admit an
identity extension lemma, which ensures that if we pass the identity
type as relations for the arguments of type constructors, then the
resulting relation for the type constructor is equivalent to the
identity.
This issue has been addressed, first on small
types~\cite{10.1145/2535838.2535852} by considering the reflexive
graph model and then on an extension of type theory with a parametric
function type~\cite{10.1145/3110276}.
In our work, we get a variant of the identity extension lemma very
easily because all the relations we consider need to be related to
equality through the coherence condition.
However, we do not exactly get identity extension: for instance
on the type of booleans \coqe{bool}, it is also possible 
to provide the relation which says that \coqe{true} is related to \coqe{false} (and dually), together with the equivalence which flips booleans. This defines a
univalent relation on \coqe{bool} because it satisfies the coherence
condition, but the relation does not coincide with equality on
booleans---it only does up to flipping booleans.

Recently, \citet{cavalloHarper:csl2020} proposed a theory
mixing both cubical type theory and parametric type theory. However,
they are not considering univalent parametricity, but rather a theory
where proofs of parametricity can make use of univalent principles, such
as functional extensionality or the univalence axiom. 

For Extensional Type Theory, \citet{krishnaswamiDreyer:csl2013} develop
an alternative view on parametricity, more in the style of Reynolds, by
giving a parametric model of the theory using quasi-PERs and a
realizability interpretation of the theory. From this model construction
and proof of the fundamental lemma they can justify adding axioms to the
theory that witness strong parametricity results, even on open
terms. However they lose the computability and effectiveness enjoyed by both
Bernardy's construction and ours, which are developed in intensional type theories.

The parametricity translation of \citet{DBLP:journals/corr/AnandM17} extends the logical relation on
propositions to force that related propositions are {\em logically}
equivalent. It can be seen as a degenerate case of our approach that
forces related types to be equivalent, considering that equivalence boils down to logical equivalence on propositions (see \S\ref{sec:prop} for a more detailed explanation).
However the translations differ in other aspects. While our
translation requires the univalence axiom, theirs assumes proof
irrelevance and the K axiom, and does not treat the type
hierarchy. Our solution to the fixpoint arising from interpreting
$\Type_i : \Type_{i+1}$ is original, along with the use of conditions
to ensure coherence with equality. They study the translation of
inductively-defined types and propositions in detail, giving specific
translations in these two cases to accommodate the elimination
restrictions on propositions, and are more fine-grained in the
assumptions necessary on relations in parametricity theorems. In both
cases, the constructions were analyzed to ensure that axioms were only
used in the non-computational parts of the translation, hence they are
effective.

\newcommand{\cat}{\mathcal{C}}
\newcommand{\Span}{\mathbf{Span}}
\newcommand{\two}{{\cdot -> \cdot}}

Relational parametricity in type theory can be understood
categorically in terms of inverse diagrams. Indeed, given a category
with attributes (CwA) $\cat$, the category of inverse diagrams
$\cat^\Span$ (where $\Span$ is the “walking Span” category :
$0 <- 01 -> 1$) is also a CwA where a type is now a triple of two types
and a relation between them. This interpretation does not extend
directly to univalent parametricity but the notion of homotopical
inverse diagrams recently developed by \citet{kapulkin2018homotopical}
seems to provide a categorical interpretation of univalent
parametricity---namely, as inverse diagrams over the category $\two$
with two objects and one arrow between them which is ``marked'' as an
equivalence. \citet{kapulkin2018homotopical} 
show that under the assumption that the CwA $\cat$ satisfies
functional extensionality, the category $\cat^\two$ is also a CwA.
However, they do not investigate the interpretation of universes and
univalent universes as done for instance by \citet{shulman2015}. Our
work suggests that the interpretation of universes for homotopical
inverse diagrams can only be done for univalent universes, in the same
way as dependent products can only be lifted in presence of functional
extensionality.

\paragraph{Data refinement.}

Another part of the literature deals with the general data refinement
problem, e.g. the ability to use different related data structures
for different purposes: typically simplicity of proofs versus efficient
computation. The frameworks provide means to systematically transport
results from one type to the other.

\citet{magaudBertot:types2000} and \citet{magaud:tphol2003} first explored the idea
of transporting proof terms from one data representation to another in
Coq, assuming that the user provides a translation of the definitions from one
datatype to the other. It is limited to isomorphism and implemented
externally as a plugin. The technique is rather invasive in the sense
that it supports the transport of proof terms that use the
computational content of the first type (\eg~the reduction rules for
\coqe{plus} on natural numbers) by making type conversions explicit,
turning them into propositional rewrite rules. This approach breaks down in
presence of type dependencies.

In CoqEAL~\cite{cohenAl:cpp2013} refinement is allowed from
proof-oriented data types to efficiency-oriented ones, relying on
generic programming for the computational part and automating the
transport of theorems and proofs. CoqEAL does not only deal with isomorphisms,
but also quotients, and even partial quotients, which we cannot
handle. The approach exploits parametricity for generating proofs, but it does not support general dependent types, only parametric polymorphism. 
Moreover, the advocated style prevents doing local
transport and rather requires working with interfaces from the outset---which we coined the {\em anticipation problem}---and applying
parametricity in a second step. We can avoid anticipating common interfaces thanks to our limitation to transport by equivalences.

In a categorical setting, \citet{ROBINSON1994163} uses parametricity
for System F to show that when a function is defined on a type $A$
using only the fact that it is isomorphic to a given type $B$, then
the function can be transported to another type $A'$ as long as it is also
isomorphic to $B$. This approach to parametricity is similar to that
of CoqEAL, which uses the fact that when working with an abstract
interface---here an abstract copy of a type---the term does not depend on
the implementation of the interface. This is different from the univalent parametricity approach we develop here, which works on two concrete types and functions defined on them, which can {\em a posteriori} be shown equivalent.

\citet{haftmannAl:itp2013} explain how the Isabelle/HOL code generator
uses data refinements to generate executable versions of abstract
programs. The refinement relation used is similar to the partial
quotients of CoqEAL.
The Autoref tool for Isabelle \cite{lammich:itp2013} also uses
parametricity for refinement-based development. It is an external tool
that synthesizes executable instances of generic algorithms and refinement
proofs.

\citet{huffmanKuncar:cpp2013} address the problem of transferring
propositions between different types, typically a representation type
(\eg~integers) to an abstract type (\eg~natural numbers) in the context
of Isabelle/HOL. Again this allows to relate a type and its quotient,
like in CoqEAL, and is based on parametricity.
Recently, \citet{zimmermannHerbelin:arxiv2015} present an algorithm and plugin to transport theorems along isomorphisms in Coq similar to that of \citet{huffmanKuncar:cpp2013}. 
In addition to requiring the user to
provide a surjective function $f$ to relate two data types, their technique demands that the user explicitly provide transfer lemmas of the form $\forall x_1 \ldots x_n,\ R(x_1\ldots x_n) \implies R'(f(x_1)\ldots f(x_n))$, for each relation $R$ that the user expects to transfer to a relation $R'$. The approach is not yet able to handle parameterized types, let alone dependent types and type-level computation.

Recently, \citet{ringer_et_al:LIPIcs:2019:11081} developed a tool in
\Coq to automatically build equivalences between inductive types using
the theory of ornaments~\cite{dagand_mcbride_2014}. These equivalences
are then instrumented to transport functions and proofs using a framework
that is largely inspired by our previous conference article on
univalent parametricity~\cite{tabareauAl:icfp2018}.  However, the approach is implemented as a \Coq plugin in \OCaml. While this can be convenient to achieve full automation without relying on the somewhat brittle typeclass mechanism, it also presents major risks by being tied to a specific implementation of \Coq. A direct implementation based on MetaCoq~\cite{sozeau:hal-02167423,MetaCoq:popl2020} would seem preferable. Independently of the implementation strategy, the tool developed by 
\citet{ringer_et_al:LIPIcs:2019:11081} is essentially the white-box transport described in this article. Because they never perform black-box (univalent) transport, however, they also encountered the computation problem of parametricity described in \S~\ref{sec:het-param}. We hope that the many clarifications provided in this extended and revised article will prove helpful for addressing these issues.

\section{Conclusion}
\label{sec:conclu}

We have presented {\em univalent parametricity}, a fruitful marriage of parametricity and univalence that fully realizes programming and proving modulo equivalences in type-theoretic proof assistants. 
Univalent parametricity supports two complementary reasoning principles and forms of transport, resulting in {\em transport {\`a} la carte}: from a type equivalence, univalent transport operates in a black-box manner; additional proofs of equivalences between functions over related types allow heterogeneous parametricity to transport terms in a white-box manner up to these equivalences.
We have shown how this makes it possible to conveniently switch between an easy-to-reason-about representation and a computationally-efficient representation. 
Our approach is realizable even in type theories where univalence is taken as an axiom, such as in \Coq, and we have explored how to maximize the effectiveness of transport in this setting. Several examples and use cases in \Coq attest to the practical impact of this work to provide easier-to-use proof assistants by supporting seamless programming and proving modulo equivalences.

\begin{acks}
We thank the reviewers of the original ICFP 2018 publication, as well as the ICFP 2018 participants for asking several questions that forced us to better understand the limits of our initial proposal. This eventually led us to the enhanced presentation of this article. We are thankful as well to Assia Mahboubi for suggesting concrete applications for validating the use of our approach, which likewise pushed us to extend and enhance this work on univalent parametricity. Finally, we are extremely grateful to the anonymous JACM reviewers for their acute and eminently helpful comments and suggestions, which have resulted in considerable clarifications and more precise expositions of the concepts. 
\end{acks}

\bibliography{biblio}

\end{document}